\documentclass[preprint,11pt,authoryear]{elsarticle}
\usepackage{xcolor}
\usepackage{amsfonts}
\usepackage{amsmath}
\usepackage{amsthm}
\usepackage{amssymb}
\usepackage{times}
\usepackage{subfig}
\usepackage{multirow}
\usepackage{setspace}
\usepackage{enumerate}
\usepackage{natbib}
\usepackage{color}
\usepackage{colortbl}
\usepackage{verbatim}
 \usepackage{lscape}
\usepackage{array}
\usepackage{caption}
\usepackage{rotating}

\usepackage{url}

\newcolumntype{x}[1]{%
>{\centering\hspace{0pt}}p{#1}}%
\definecolor{Gray}{gray}{0.9}
\newcolumntype{g}{>{\columncolor{Gray}}c}

\def\iid{\buildrel {\rm i.i.d.} \over \sim}

\def\i.i.d.{\buildrel {\rm i.i.d.} \over \sim}

\def\cw#1 { \overset{\mathbb{P}}{\underset{#1}{\longrightarrow}} }
\def\Real{\mathbb{R}}
\def\Natu0{\mathbb{N}_0}
\def\P#1{{\mathbb{P}}\left(#1\right)}

\def\E#1{{\mathbb E}\left[#1\right]}

\def\Var#1{{\mathrm Var}\left(#1\right)}

\def \rcov#1#2 {{\rm cov}_{#1}\left( #2\right)}

\DeclareMathOperator*{\argmin}{arg\,min}

\newtheorem{example}{Example}

\newtheorem{theorem}{Theorem}

\newtheorem{corollary}{Corollary}
\newtheorem{remark}{Remark}
\newtheorem{proposition}{Proposition}
\newtheorem{assumption}{Assumption}
\newtheorem*{toy*}{Toy Model}

\def\cov#1{{\rm  cov}\left[#1\right]}

\begin{document}
\begin{frontmatter}
\title{Implementation of Frequency-Severity Association in BMS Ratemaking}


\author[EH]{Rosy Oh}
\ead{rosy.oh5@gmail.com}

\author[WI]{Peng Shi\corref{cor2}}
\ead{pshi@bus.wisc.edu}

\author[EH]{Jae Youn Ahn\corref{cor2}}
\ead{jaeyahn@ewha.ac.kr}
\address[EH]{Department of Statistics, Ewha Womans University, Seodaemun-Gu, Seoul 03760, Korea.}
\address[WI]{Wisconsin School of Business, University of Wisconsin-Madison, Madison, WI 53706, USA.}
\cortext[cor2]{Corresponding Authors}

\begin{abstract}
A Bonus-Malus System (BMS) in insurance is a premium adjustment mechanism widely used in a
posteriori ratemaking process to set the premium for the next contract period based on a policyholder's claim history. The current practice in BMS implementation relies on the assumption of independence
between claim frequency and severity, despite the fact that a series of recent studies report evidence of a
significant frequency-severity relationship, particularly in automobile insurance. To address this discrepancy, we propose
a copula-based correlated random effects model to accommodate the dependence between claim frequency
and severity, and further illustrate how to incorporate such dependence into the current BMS. We derive analytical
solutions to the optimal relativities under the proposed framework and provide numerical experiments based on real data analysis to assess the effect of frequency-severity dependence in BMS ratemaking.

\end{abstract}


\end{frontmatter}

\vfill

\pagebreak

\vfill

\pagebreak

\section{Introduction}
A Bonus-Malus System (BMS) is an a posteriori ratemaking mechanism to set a premium based on a policyholder's claim history.
While a BMS could take various forms across different regions, all share three major features.
First, a policyholder is classified into a risk class, or a BMS level, according to his or her past claim history.
Second, at each policy renewal, the BMS level of the policyholder either stays at the current level or migrates to another level according to pre-specified transition rules based on the claim experience in the current policy year.
Third, the premium is linked directly to the BMS level via a premium adjustment coefficient known as the Bonus-Malus (BM) relativity. Specifically, the premium is adjusted as the product of BM relativity and the base premium, which is determined by the policyholder's risk characteristics. Under the assumption of independence between claim frequency and severity, one of the key assumptions in the classical BMS literature,
the base premium is the product of the mean frequency and mean severity \citep{Pitrebois2003, Denuit2, Chong}.

%

However, recent insurance studies report significant dependence between claim frequency and severity. To accommodate such dependence, prior works propose various statistical models. Examples include copula models \citep{Czado, Frees4}, shared or bivariate random effects models \citep{Cossette, Bastida, Czado2015}, and two-part models \citep{Peng, Garrido, Park2018does, AhnValdez2}.
Despite the aforementioned models, the frequency-severity dependence has not been explicitly examined in the BMS context. On one hand, it is arguable that if claim frequency and severity are correlated, then it is critical to incorporate this relationship in the BMS implementation. This is because if claim frequency is used in the current BMS to indicate the latent risk of a policyholder, then so is claim severity. On the other hand, implementing this dependence in the current implementation of BMS is challenging. First, the posterior premium depends on both the base premium and BMS relativity. The frequency-severity correlation will affect both components in the posterior risk adjustment. Second, the current BMS transition rule is determined solely by the history of claim frequency. Thus, it is not straightforward to incorporate the frequency-severity association in the current BMS framework.


Alternatively, the literature suggests various forms of modified BMS to accommodate this dependence \citep{Pinquet, Gomez}. However, the risk adjustment process in the proposed BMS inevitably involve both claim frequency and severity, which result in an increase in the complexity in the a posteriori ratemaking system.

The goal of our study is to propose a simple yet flexible framework to incorporate the frequency-severity association into the posterior risk adjustment under the current form of BMS. Specifically, we employ a copula-based bivariate random effects model to accommodate the dependence between claim frequency and severity. Using the proposed model, we further derive the optimal BM relativity in which the transition rule depends on only the claim frequency, and yet, the BMS can capture the association between frequency and severity.
%

The remainder of this paper is organized as follows.
Section \ref{bm.sec} presents the basic definitions and notations.
In Section \ref{sec.3}, we propose a copula-based random effects model for the general frequency-severity model. Section \ref{sec.4} shows that, through a simple modification of the optimal relativities, the current BMS
can properly accommodate the frequency-severity dependence. Section \ref{sec.6} explains the impact of the dependence on the relativities through a real data analysis and numerical study,
followed by concluding remarks in Section \ref{sec.7}.

\section{Definitions and Notations}\label{bm.sec}

 We consider a portfolio of policyholders in the context of short-term insurance, where a policyholder could decide whether or not to renew the policy at the end of each policy year and the insurer can adjust the premium at the beginning of each policy year based on the policyholder's claim experience. We denote $\mathbb{N}$, $\mathbb{N}_0$, $\Real$, and $\Real^+$ by the set of natural numbers, the set of non-negative integers, the set of real numbers, and the set of positive real numbers, respectively.
Define $N_{it}$ to indicate the number of claims of the $i$-th policyholder in the $t$-th policy year, and
\[
\boldsymbol{Y}_{it}:=\begin{cases}
  (Y_{it, 1}, \cdots, Y_{it, N_{it}}) & N_{it}>0\\
  \hbox{Not defined} & N_{it}=0
\end{cases}
\]
denotes the associated individual claim amount. We further define the \textit{aggregate claim amount} and average claim amount as
\[
S_{it}:=\begin{cases}
\sum\limits_{j=1}^{N_{it}}Y_{it,j} & N_{it}>0\\
0 &N_{it}=0
\end{cases}
\quad\quad
\hbox{and}
\quad\quad
M_{it}:=\begin{cases}
\frac{\sum\limits_{j=1}^{N_{it}}Y_{it,j}}{N_{it}} & N_{it}>0\\
\hbox{0} &N_{it}=0
\end{cases},
\]
respectively, and they satisfy the relationship
\[
S_{it}=N_{it}M_{it}.
\]
We denote the realizations of $N_{it}$, $\boldsymbol{Y}_{it}$, $S_{it}$, and $M_{it}$ by $n_{it}$, $\boldsymbol{y}_{it}$, $s_{it}$, and $m_{it}$, respectively. The actuarial science literature often refers to $N_{it}$ as the \textit{frequency}, $\boldsymbol{Y}_{it}$ as the \textit{individual severity}, and $M_{it}$ as the \textit{average severity} of the insurance claims. Furthermore, we define the claim history of the $i$-th policyholder at the end of year $T$ as:
\[
\mathcal{F}_{it}:=\left\{(n_{i1}, m_{i1} ), \cdots, (n_{it}, m_{it} )
\right\}.
\]
We emphasize that our proposed method requires only information on the average severity, and imposes no constraints on individual severity. In the following text, we refer to the model for $(N_{it}, M_{it})$
as the \textit{frequency-severity model}.

In the BMS context, we differentiate between \textit{a priori} and \textit{a posterior} ratemaking. The former is based on observed policyholders' risk characteristics and the latter is based on policyholders' unobserved risk characteristics. Let $\left(\boldsymbol{X}_i^{[1]},\boldsymbol{X}_i^{[2]}\right)$ and $\left(\Theta_i^{[1]},\Theta_i^{[2]}\right)$ denote, respectively, the observed and unobserved risk characteristics for the $i$-th policyholder.
Note that $\left(\boldsymbol{X}_i^{[1]},\boldsymbol{X}_i^{[2]}\right)$, and $\left(\Theta_i^{[1]},\Theta_i^{[2]}\right)$ do not have the subscript $t$, as we assume that they are constant in time.
The superscripts [1] and [2] are indices for the frequency and severity components, respectively.

To model the frequency and severity of insurance claims, we follow \cite{deJong2008} and use generalized linear models (GLMs). Specifically, we consider the {\it exponential dispersion family} (EDF) in \citet{Nelder1989}. The EDF with mean $\mu$ and dispersion $\psi$, denoted by $F(\cdot;\mu, \psi)$, has the probability density/mass function
\[
p(y|\theta,\psi)=\exp\left[(y\theta-b(\theta))/\psi \right]c(y,\psi),
\]
where $\theta$ is the canonical parameter, and $b(\cdot)$ and $c(\cdot)$ are predetermined functions. The mean of the random variable can be shown to be $$\mu=\E{Y}=b^\prime(\theta),$$
and the variance is calculated as $$\Var{Y}=b^{\prime\prime}(\theta)\psi\equiv V(\mu)\psi,$$ where the inverse of $b^{'}(\cdot)$ is the link function, and $V(\cdot)$ is the variance function.

\section{Copula-based Correlated Random Effects Model}\label{sec.3}

This section presents the general frequency-severity modeling framework. We start with a homogeneous model that does not allow belief updating for policyholders' risk profile, and then we extend the model to allow for unobserved subject-specific effects. Suppose that the insurer predetermines $\mathcal{K}$ risk classes based on the policyholders' risk characteristics. Let $\left(\boldsymbol{x}_{\kappa}^{[1]},\boldsymbol{x}_{\kappa}^{[2]}\right)$ define the $\kappa$-th risk class, and $w_\kappa$ be the weight of the risk class; then,
  \begin{equation}\label{eq.1}
  w_\kappa:= \P{\boldsymbol{X}_{i}^{[1]}=\boldsymbol{x}_{\kappa}^{[1]},
  \boldsymbol{X}_{i}^{[2]}=\boldsymbol{x}_{\kappa}^{[2]}
  }, \quad \kappa = 1,\cdots, \mathcal{K}.
  \end{equation}
Denote $\left(\Lambda_i^{[1]},\Lambda_i^{[2]}\right)$ as the a priori premium for the $i$-th policyholder, which is determined by the observed risk characteristics as follows:
\[
        \Lambda_{i}^{[1]}=\eta_1^{-1} \left(\boldsymbol{X}^{[1]}_{i}\boldsymbol{\beta}^{[1]}\right)
        \quad\hbox{and}\quad \Lambda_{i}^{[2]}=\eta_2^{-1}\left(\boldsymbol{X}^{[2]}_{i}\boldsymbol{\beta}^{[2]}\right),
  \]
where $\eta_1(\cdot)$ and $\eta_2(\cdot)$ are link functions, and $\boldsymbol{\beta}^{[1]}$ and $\boldsymbol{\beta}^{[2]}$ are the parameters to estimate.

\bigskip
{\bf Homogeneous Model.}
Conditioning on the observed risk characteristics, we construct the frequency-severity model for $(N_{it},M_{it})$ as follows:
\begin{enumerate}
    \item[i.] We specify the frequency component $N_{it}$
    conditional on the observed risk characteristics
    $$\left(\boldsymbol{X}_{i}^{[1]}, \boldsymbol{X}_{i}^{[2]}\right)=\left(\boldsymbol{x}_{i}^{[1]}, \boldsymbol{x}_{i}^{[2]}\right)$$
    using a count regression model, such that
        \[
        N_{it}|\left(\boldsymbol{x}_{i}^{[1]},\boldsymbol{x}_{i}^{[2]}\right) \sim {F}_1\left(\cdot ; \lambda_{i}^{[1]}, \psi^{[1]} \right), \quad\hbox{with}\quad \lambda_{i}^{[1]}= \eta_1^{-1}\left(\boldsymbol{x}_{i}^{[1]}\boldsymbol{\beta}^{[1]}\right),
        \]
    where the distribution $F_1$ has mean parameter $\lambda_{i}^{[1]}$ and the other parameter is $\psi^{[1]}$.
For $F_1$ in EDF, $\psi^{[1]}$ is the dispersion parameter.

   \item[ii.] For severity, we specify the conditional distribution of the average severity conditional on the number of claims $N_{it}=n_{it}$ and the observed risk characteristics $$\left(\boldsymbol{X}_{i}^{[1]}, \boldsymbol{X}_{i}^{[2]}\right)=\left(\boldsymbol{x}_{i}^{[1]}, \boldsymbol{x}_{i}^{[2]}\right).$$ Specifically, we assume
  \[
       M_{it}\big\vert \left(n_{it}, \boldsymbol{x}_{i}^{[1]}, \boldsymbol{x}_{i}^{[2]}\right) \iid {F}_2\left(\cdot\,;\, \lambda_{i}^{[2]}, \psi^{[2]}/n_{it}\right), \quad n_{it}>0
  \]
and
  \[
       \P{M_{it}=0\big\vert N_{it}=n_{it}, \boldsymbol{x}_{i}^{[1]},\boldsymbol{x}_{i}^{[2]}}=1, \quad  n_{it}=0,
  \]
   where the distribution $F_2$ in EDF has mean parameter $ \lambda_{i}^{[2]}$ 
    with $\lambda_{i}^{[2]}= \eta_2^{-1}\left(\boldsymbol{x}_{i}^{[2]} \boldsymbol{\beta}^{[2]} \right)$,
   and dispersion parameter $\psi^{[2]}/n_{it}$. This specification follows \citet{Frees4} and \citet{Garrido}, where the focus is to introduce the dependence between frequency and severity in the context of GLMs.

\end{enumerate}
Note that the homogeneous model allows for observed heterogeneity and the dependence between the frequency and severity components is accommodated through conditional distributions. Furthermore, given the policyholders' observed risk characteristics, we assume that the vector $(N_{it},S_{it})$ is independent across subjects and across time. Thus, there is no experience learning cross-sectionally and intertemporally.

\bigskip
{\bf Heterogeneous Model.} Presumably, insurers' risk classifications are not complete in that there is unobserved heterogeneity after accounting for policyholders' observed risk characteristics. Further, such unobserved heterogeneity can be learned over time with repeated observations. To accommodate such unobserved risk characteristics, we consider the following correlated random effects model:

\begin{enumerate}
    \item[i.] We specify the frequency component $N_{it}$ using a count regression model conditional on the observed and unobserved risk characteristics
        $$\left(\boldsymbol{X}_{i}^{[1]}, \boldsymbol{X}_{i}^{[2]}\right)=\left(\boldsymbol{x}_{i}^{[1]}, \boldsymbol{x}_{i}^{[2]}\right)
        \quad\hbox{and}\quad \left(\Theta_i^{[1]}, \Theta_i^{[2]}\right)
    =\left(\theta_i^{[1]}, \theta_i^{[2]}\right),
        $$  
    such that
       \begin{equation}\label{eq.n}
        N_{it}\big\vert \left(\boldsymbol{x}_{i}^{[1]},\boldsymbol{x}_{i}^{[2]}, \theta_i^{[1]},
        \theta_i^{[2]}\right) \iid {F}_1\left(\cdot ; \lambda_{i}^{[1]}\theta_i^{[1]}, \psi^{[1]} \right), 
        \end{equation}
   where the distribution $F_1$ has a mean parameter of $\lambda_{i}^{[1]}\theta_i^{[1]}$ with $\lambda_{i}^{[1]}= \eta_1^{-1}\left(\boldsymbol{x}_{i}^{[1]}\boldsymbol{\beta}^{[1]}\right)$ and the other parameter is $\psi^{[1]}$. For $F_1$ in EDF, $\psi^{[1]}$ is the dispersion parameter.

   \item[ii.] The conditional distribution of average severity conditional on
the frequency $N_{it}=n_{it}$ and the observed and unobserved risk characteristics
        $$\left(\boldsymbol{X}_{i}^{[1]}, \boldsymbol{X}_{i}^{[2]}\right)=\left(\boldsymbol{x}_{i}^{[1]}, \boldsymbol{x}_{i}^{[2]}\right)
        \quad\hbox{and}\quad \left(\Theta_i^{[1]}, \Theta_i^{[2]}\right)
    =\left(\theta_i^{[1]}, \theta_i^{[2]}\right),
        $$
   is given by
  \begin{equation}\label{eq.s}
       M_{it}\big\vert \left(n_{it}, \boldsymbol{x}_{i}^{[1]},\boldsymbol{x}_{i}^{[2]}, \theta_i^{[1]},
        \theta_i^{[2]}\right) \iid {F}_2\left(\cdot\,;\, \lambda_{i}^{[2]}\theta_i^{[2]}, \psi^{[2]}/n_{it}\right), \quad n_{it}>0
  \end{equation}
 and
  \[
       \P{M_{it}=0\big\vert n_{it}, \boldsymbol{x}_{i}^{[1]},\boldsymbol{x}_{i}^{[2]}, \theta_i^{[1]},
        \theta_i^{[2]}}=1, \quad  n_{it}=0,
  \]
   where the distribution $F_2$ in the EDF has mean parameter $\lambda_{i}^{[2]}\theta_i^{[2]}$ 
     with $\lambda_{i}^{[2]}=\eta_2^{-1}\left(\boldsymbol{x}_{i}^{[2]} \boldsymbol{\beta}^{[2]}\right)$
   and dispersion parameter $\psi^{[2]}/n_{it}$.

   \item[iii.] The unobserved risk characteristic $\left(\Theta_i^{[1]},\Theta_i^{[2]}\right)$, which we assume is independent of the observed risk characteristic
   $\left(\boldsymbol{X}_{i}^{[1]}, \boldsymbol{X}_{i}^{[2]} \right)$, has a joint distribution, $H$, of
   \begin{equation}\label{eq.4}
  \left(\Theta_i^{[1]}, \Theta_i^{[2]} \right)\iid H=C\left(G_1, G_2 \right), \quad i=1, \cdots, I,
  \end{equation}
  where $G_1$ and $G_2$ denote the marginal distributions for $\Theta_i^{[1]}$ and $\Theta_i^{[2]}$, respectively. The function $C$ is the bivariate copula that defines the joint distribution $H$. We use $g_1$, $g_2$, and $h$ to denote the density versions of $G_1$, $G_2$, and $H$, respectively.
Further, we assume
\[
       \E{\Theta_i^{[1]}} = 1
        \quad\hbox{and}\quad \E{\Theta_i^{[2]}} = 1,
  \]
for identification purposes.

\end{enumerate}

The Heterogeneous Model belongs to the family of correlated random effects models in which the unobserved risk characteristics are treated as random effects.
Given the observed risk characteristics and the unobserved risk, we assume that the claim history
  $  \left(N_{it}, M_{it}\right)  $
  is independent across subjects and across time.
Different from the homogeneous model, the subject-specific effects introduce several types of dependence in the context of frequency-severity models, and thus allow insurers to update information on policyholders' risk profile. First, the random effect $\Theta_i^{[1]}$ induces serial correlation in the frequency component and allows insurers to learn the hidden risks affecting the likelihood of a claim. Second, the random effect $\Theta_i^{[2]}$ induces serial correlation in the severity component, which allows insurers to obtain additional information on the hidden risks given claim frequency. Third, the dependence between $\Theta_i^{[1]}$ and $\Theta_i^{[2]}$ accommodates the relationship between frequency and severity, which, as we show, has a critical effect on BMS ratemaking.

As opposed to the conditional probability model in \citet{Frees4} and \citet{Garrido}, the random effects framework is an alternative approach to modeling the dependence between frequency and severity in insurance claims. For example, \citet{Bastida} employ correlated random effects to introduce the association between frequency and individual severity, and \citet{Czado2015} and \citet{JKWoo} consider a special case of a correlated random effects model to allow for the association between frequency and average severity. We adopt a correlated random effects framework in a setting different from these existing studies. First, we use the correlated random effects in a longitudinal context, with the purpose of not only introducing the frequency-severity association, but also accounting for unobserved heterogeneity, which allows the insurer to perform experience rating. Second, in terms of the frequency-severity dependence, our model takes a hybrid approach to the conditional probability model and the random effects model, where the relationship between frequency and severity is accommodated by both the time-constant latent risk and the conditional distribution of the average severity given the frequency.

%

In the following proposition, we provide the expectation formulas conditioned on a priori information under the Heterogeneous Model.
 The conditional expressions are of primary interest to insurers because a priori information is usually available at the time of the contract. For brevity, we omit subscript $i$ in the following text.
 Note that while the frequency distribution ${F}_1$ can be any counting distribution, including zero-inflated distributions, the following proposition assumes that $F_1$ is in the EDF for a simple calculation.

  \begin{proposition}\label{prop.1}
 Under the Heterogeneous Model, we assume that the frequency distribution function ${F}_1$ is in ${\it EDF}(\lambda^{[1]}\Theta^{[1]}, \psi^{[1]})$.
    Furthermore,
  let ${V}_1$ and ${V}_2$ be the variance functions for the frequency distribution $F_1$ and severity distribution $F_2$, respectively. Then, we have the following conditional expressions. 

 \begin{enumerate}
 \item[i.] The mean and variance of the aggregate severity a priori are
\[
\E{S_{t}\big\vert \boldsymbol{x}^{[1]}, \boldsymbol{x}^{[2]}}
=\lambda^{[1]} {\lambda^{[2]}}\E{\Theta^{[1]}\Theta^{[2]} }\\
\]
and
\begin{equation}\label{eq.201}
\begin{aligned}
\Var{S_{t}\bigg\vert \boldsymbol{x}^{[1]}, \boldsymbol{x}^{[2]}}
&=\E{ \lambda^{[1]}\Theta^{[1]} V_2\left(\lambda^{[2]}\Theta^{[2]} \right) \psi^{[2]}
\Bigg\vert \boldsymbol{x}^{[1]}, \boldsymbol{x}^{[2]} }+ \left(\lambda^{[1]} \lambda^{[2]}\right)^2 \Var{\Theta^{[1]}\Theta^{[2]} }\\
 &
 +\left( \lambda^{[2]} \right)^2 \E{\left( \Theta^{[2]} \right)^2 V_1\left(\lambda^{[1]} \Theta^{[1]}\right) \psi^{[1]}\Bigg\vert \boldsymbol{x}^{[1]}, \boldsymbol{x}^{[2]}}.
 \\
\end{aligned}
\end{equation}
\item[ii.] The covariances related to severity a priori are
\[
\cov{S_{t_1}, S_{t_2}\bigg\vert \boldsymbol{x}^{[1]}, \boldsymbol{x}^{[2]}}=\left(\lambda^{[1]}\lambda^{[2]}\right)^2\Var{\Theta^{[1]}\Theta^{[2]}}
\]
and
\[
\cov{N_{t_1}, S_{t_2}\bigg\vert \boldsymbol{x}^{[1]}, \boldsymbol{x}^{[2]}}=
\left(\lambda^{[1]}\right)^2\lambda^{[2]}\E{\left(\Theta^{[1]}\right)^2\Theta^{[2]}}
-\left(\lambda^{[1]}\right)^2\lambda^{[2]}
\E{\Theta^{[1]}\Theta^{[2]}}\\
\]
{for $t_1\neq t_2$, and}
\[
\begin{aligned}
\cov{N_{t}, S_{t}\bigg\vert \boldsymbol{x}^{[1]}, \boldsymbol{x}^{[2]}}&=
\lambda^{[2]} \psi^{[1]}\E{{V}_1\left( \lambda^{[1]} \Theta^{[1]}\right) \Theta^{[2]}} +
\left(\lambda^{[1]}\right)^2\lambda^{[2]}\E{\left(\Theta^{[1]}\right)^2\Theta^{[2]}}\\
&\quad\quad\quad\quad-\left(\lambda^{[1]}\right)^2\lambda^{[2]}
\E{\Theta^{[1]}\Theta^{[2]}}.\\
\end{aligned}
\]
\item[iii.] The covariance between the frequencies a priori is
\[
\cov{N_{t_1}, N_{t_2}\bigg\vert \boldsymbol{x}^{[1]}, \boldsymbol{x}^{[2]}}=\left(\lambda^{[1]}\right)^2\Var{\Theta^{[1]}}
\]
for $t_1\neq t_2$.

\end{enumerate}

\end{proposition}

While the variance formula in \eqref{eq.201} is complicated, the conditional expression in Proposition \ref{prop.1} can be further simplified for the most exponential family distribution of the severity, as the following example shows.


\begin{example}\label{ex.prop.1}
Under the Heterogeneous Model, 
consider the Poisson distribution for $F_1$ and Gamma distribution for $F_2$, with the log link functions for $\eta_1$ and $\eta_2$, respectively:
\[
     N_{t}\big\vert \left(\boldsymbol{x}^{[1]},\boldsymbol{x}^{[2]}, \theta^{[1]},
        \theta^{[2]}\right) \sim {\rm Pois}\left(\lambda^{[1]}\theta^{[1]}\right),
\]
\[
   M_{t}\big\vert \left(n_{t}, \boldsymbol{x}^{[1]},\boldsymbol{x}^{[2]}, \theta^{[1]},
        \theta^{[2]}\right) \iid {\rm Gamma} \left( \lambda^{[2]}\theta^{[2]}, \psi^{[2]}/n_{t}\right), \quad n_{t}>0,
\]
and
  \[
       \P{M_{t}=0\big\vert n_{t}, \theta^{[1]},\boldsymbol{x}^{[1]},
     \theta^{[2]},\boldsymbol{x}^{[2]}}=1, \quad  n_{t}=0,
  \]
where
\[
\lambda^{[1]}=\exp\left(\boldsymbol{x}^{[1]}\boldsymbol{\beta}^{[1]}\right) \quad\hbox{and}\quad
\lambda^{[2]}=\exp\left(\boldsymbol{x}^{[2]} \boldsymbol{\beta}^{[2]}\right).
\]

Furthermore, assume that
the Gaussian copula with correlation coefficient $\rho$ is the distribution of the residual effect $\left( \Theta^{[1]}, \Theta^{[2]}\right)$, with
${\rm Log Normal}(-\sigma_1^2/2, \sigma_1^2)$ and ${\rm Log Normal}(-\sigma_2^2/2, \sigma_2^2)$
as the marginal distributions of $\Theta^{[1]}$ and $\Theta^{[2]}$, respectively.
Then, from Proposition \ref{prop.1}, we have the following covariance formulas:
\[
\begin{cases}
  {\rm cov}\left(N_{t_1},N_{t_2}\bigg\vert \boldsymbol{x}^{[1]}, \boldsymbol{x}^{[2]}\right)=\left(\lambda^{[1]}\right)^2\left(\exp(\sigma_1^2)-1\right);\\
{\rm cov}\left(S_{t_1},S_{t_2}\bigg\vert \boldsymbol{x}^{[1]}, \boldsymbol{x}^{[2]}\right)
=\left(\lambda^{[1]}\lambda^{[2]}\right)^2 \left\{\exp(\sigma_1^2+\sigma_2^2+4\rho\sigma_1\sigma_2)-\exp(2\rho\sigma_1\sigma_2)\right\};\\
  {\rm cov}\left(N_{t_1},S_{t_2}\bigg\vert \boldsymbol{x}^{[1]}, \boldsymbol{x}^{[2]}\right)
=\left(\lambda^{[1]}\right)^2 \lambda^{[2]}\left\{\exp(\sigma_1^2+2\rho\sigma_1\sigma_2)-\exp(\rho\sigma_1\sigma_2)\right\} ;\\
\end{cases}
\]
for $t_1\neq t_2$, and
\[
  {\rm cov}\left(N_{t},S_{t}\bigg\vert \boldsymbol{x}^{[1]}, \boldsymbol{x}^{[2]}\right)
=\lambda^{[1]} \lambda^{[2]}\exp(\rho\sigma_1\sigma_2) +
\left(\lambda^{[1]}\right)^2 \lambda^{[2]}\left\{\exp(\sigma_1^2+2\rho\sigma_1\sigma_2)-\exp(\rho\sigma_1\sigma_2)\right\}.
\]
Furthermore, using the covariance formulas above and the variance formulas
\[
\begin{aligned}
\Var{N_{t}\bigg\vert \boldsymbol{x}^{[1]}, \boldsymbol{x}^{[2]}}
&= \left(\lambda^{[1]}\right)^2 \left\{\exp(\sigma_1^2)-1\right\} +\lambda^{[1]}
\end{aligned}
\]
and
\[
\begin{aligned}
\Var{S_{t}\bigg\vert \boldsymbol{x}^{[1]}, \boldsymbol{x}^{[2]}}
&= \lambda^{[1]}\left(\lambda^{[2]}\right)^2 \psi^{[2]}\exp(\sigma_2^2+2\rho\sigma_1\sigma_2) \\
&\quad\quad\quad+
\left(\lambda^{[1]}\lambda^{[2]}\right)^2 \left\{\exp(\sigma_1^2+\sigma_2^2+4\rho\sigma_1\sigma_2)-\exp(2\rho\sigma_1\sigma_2)\right\},
\end{aligned}
\]
we have the following conditional correlations:
\[
\begin{cases}
  {\rm corr}\left(N_{t_1},N_{t_2}\bigg\vert \boldsymbol{x}^{[1]}, \boldsymbol{x}^{[2]}\right)
=\frac{\lambda^{[1]}\left\{\exp(\sigma_1^2)-1\right\}}{1+\lambda^{[1]}\left\{\exp(\sigma_1^2)-1\right\}};\\
{\rm corr}\left(S_{t_1},S_{t_2}\bigg\vert \boldsymbol{x}^{[1]}, \boldsymbol{x}^{[2]}\right)
=\frac{
\lambda^{[1]} \left\{\exp(\sigma_1^2+\sigma_2^2+4\rho\sigma_1\sigma_2)-\exp(2\rho\sigma_1\sigma_2)\right\}
}{
 \psi^{[2]}\exp(\sigma_2^2+2\rho\sigma_1\sigma_2) +
\lambda^{[1]} \left\{\exp(\sigma_1^2+\sigma_2^2+4\rho\sigma_1\sigma_2)-\exp(2\rho\sigma_1\sigma_2)\right\}
};\\
{\rm corr}\left(N_{t_1},S_{t_2}\bigg\vert \boldsymbol{x}^{[1]}, \boldsymbol{x}^{[2]}\right)
=\frac{
\lambda^{[1]} \left\{\exp(\sigma_1^2+2\rho\sigma_1\sigma_2)-\exp(\rho\sigma_1\sigma_2)\right\}
}{\sqrt{1+\lambda^{[1]}\left\{\exp(\sigma_1^2)-1\right\}}
\sqrt{
\psi^{[2]}\exp(\sigma_2^2+2\rho\sigma_1\sigma_2) +
\lambda^{[1]} \left\{\exp(\sigma_1^2+\sigma_2^2+4\rho\sigma_1\sigma_2)-\exp(2\rho\sigma_1\sigma_2)\right\}
}
} ;\\
\end{cases}
\]
for $t_1\neq t_2$, and
\[
\begin{aligned}
&{\rm corr}\left(N_{t},S_{t}\bigg\vert \boldsymbol{x}^{[1]}, \boldsymbol{x}^{[2]}\right)\\
&=\frac{
 \exp(\rho\sigma_1\sigma_2) +
\lambda^{[1]} \left\{\exp(\sigma_1^2+2\rho\sigma_1\sigma_2)-\exp(\rho\sigma_1\sigma_2)\right\}
}{\sqrt{1+\lambda^{[1]}\left\{\exp(\sigma_1^2)-1\right\}}
\sqrt{
\psi^{[2]}\exp(\sigma_2^2+2\rho\sigma_1\sigma_2) +
\lambda^{[1]} \left\{\exp(\sigma_1^2+\sigma_2^2+4\rho\sigma_1\sigma_2)-\exp(2\rho\sigma_1\sigma_2)\right\}
}
}.
\end{aligned}
\]
%
%
\end{example}

\section{Optimal Relativities in a BMS Assuming Dependence}\label{sec.4}

Recent studies propose and study models considering serial dependence as well as cross sectional dependence between the claim frequency and severity \citep{Pinquet, Bastida, Czado, Czado2015, Frees4, Gomez,AhnValdez2, JKWoo}.
Naturally, all of these proposed models require the use of past experience of both the frequencies and severities for the posteriori classification.

While predictions using such dependence models are superior, their usage in practice can be limited under the classic BMS. The main difficulty is the discrepancy between the assumptions in the classical BMS and the dependent frequency-severity models: the classic BMS implicitly assumes independence between frequency and severity, while the aforementioned studies relax such dependence.
Specifically, the classical BMS only updates the claim frequency, but not severity, and the BMS premium is calculated as the product of the posterior frequency and average severity.
Consequently, the current BMS does not seem to have room for the use of the frequency-severity dependence information in the ratemaking system.
Due to this discrepancy, the dependence models in the literature adapted a general credibility system besides the classical BMS.
For example, \citet{Pinquet} proposes using the predictive mean of aggregate severity based on the Bayesian estimation method. As another example, to accommodate the frequency-severity dependence, \citet{Gomez} proposes modifying the BMS using a transition rule that distinguishes the size of the severity.

In this section, we show that through a simple modification of the optimal relativity $r_l$, the current BMS can properly accommodate the frequency-severity dependence.
Specifically, we propose the optimal relativities under a new objective function that directly incorporate the frequency-severity dependence.
Consequently, in the proposed method, we have the same transition rules as in the classical BMS, while including the modified relativities to account for the frequency-severity dependence.

\subsection{Structure of the Classical BMS}
A BMS is a premium adjustment mechanism widely used in the a posteriori ratemaking process to set the proper premium based on the policyholder's claim history. Apart from peripheral differences, all BMS consist of BMS levels, transition rules, and relativities. We assume that there are $z$ BM levels labeled from $1$ to $z$. We use the transition rule commonly known as the $-1/+h$ system, where
each reported claim is penalized by climbing $h$ BM levels per claim, while a claim-free year is rewarded by dropping one BM level.
Since the current BM level and the number of reported claims are enough to determine the BM level in the subsequent year, the process we describe for the BM levels is a Markov chain on a finite space.

We further denote the random variable $L$ as the BM level occupied by the randomly selected policyholder in the steady state. Since the BM levels under the classical BMS are determined not by the claim severity, but by the claim frequency, the stationary distribution $\P{L=l}$ for $l=1, \cdots, z$ is determined by the information on the frequency. Specifically, under the Heterogeneous Model, we have
\begin{equation*}
\begin{aligned}
\P{L=l}
&=\sum\limits_{\kappa \in\mathcal{K}} w_{\kappa} \int \pi_l\left(\lambda_{\kappa}^{[1]}\theta^{[1]}, \psi_{\kappa}^{[1]} \right) h(\theta^{[1]}, \theta^{[2]}){\rm d}\theta^{[1]}{\rm d}\theta^{[2]}\\
&=\sum\limits_{\kappa \in\mathcal{K}} w_{\kappa} \int \pi_l\left(\lambda_{\kappa}^{[1]}\theta^{[1]}, \psi^{[1]} \right) g_1(\theta^{[1]}){\rm d}\theta^{[1]}, \quad \hbox{for\quad$l=1, \cdots, {z}$},
\end{aligned}
\end{equation*}
where $\pi_l\left(\lambda_{\kappa}^{[1]}\theta^{[1]}, \psi^{[1]} \right)$ is the stationary distribution for the policyholder with $\Lambda^{[1]}=\lambda^{[1]}_{\kappa}$ and $\Theta^{[1]}=\theta^{[1]}$ to be at level $l$.
 Note that $\pi_l$ is a conditional stationary distribution; hence, it is free of risk characteristics for severity.

We denote the relativity associated with BM level $l$ as $r_l$. The policyholder at BM level $l$ pays the product of $r_l$ and the a priori premium, which is
predetermined based on the policyholder's characteristics. Typically, the relativity $r_l$ is determined as the solution of the optimization problem.
In the following subsection, we explain how to determine the relativity $r_l$ in the optimization problem setting.

\subsection{Determination of Optimal Relativities Assuming Dependence}

In this subsection, 
we discuss the procedure to determine the optimal relativities 
in the Heterogeneous Model.
Motivated by the following Bayesian premiums
\begin{equation*}
\E{S_{t+1}\big\vert \lambda^{[1]},\lambda^{[2]}, \mathcal{F}_t} =\lambda^{[1]} \lambda^{[2]} \E{\Theta^{[1]}\Theta^{[2]}\big\vert \lambda^{[1]},\lambda^{[2]},  \mathcal{F}_t},
\end{equation*}
we set the premium using
\[
\lambda^{[1]}\lambda^{[2]}r_l
\]
for a policyholder at BM level $l$ with the given a priori
\[
\left(\Lambda^{[1]}, \Lambda^{[2]}\right) = \left(\lambda^{[1]}, \lambda^{[2]}\right).
\]
Furthermore, the optimal relativities $r_1, \cdots,$ and $r_z$ are defined as the solution to the following optimization problem
 \begin{equation}\label{eq.67}
 \begin{aligned}
(r_1^{\rm Dep}, \cdots, r_z^{\rm Dep})&:=\argmin\limits_{(r_1, \cdots, r_{{z}})\in \Real^{{z}}}\E{\left(\E{S_{t+1}\big\vert \Lambda^{[1]},  \Lambda^{[2]}, \Theta^{[1]}, \Theta^{[2]}}-\Lambda^{[1]} \Lambda^{[2]} r_L\right)^2}\\
&=\argmin\limits_{(r_1, \cdots, r_{{z}})\in \Real^{{z}}}\E{ \left(\Lambda^{[1]} \Lambda^{[2]}\Theta^{[1]}\Theta^{[2]}- \Lambda^{[1]} \Lambda^{[2]} r_L\right)^2}\\
\end{aligned}
 \end{equation}
Note that similar to \citet{Chong}, the objective function in \eqref{eq.67} is interpreted as
the expected square difference between the aggregate severity and the premium determined by BMS, and can be analytically expressed as in the following theorem.

\begin{theorem}\label{thm.oh.1}
  Under the Heterogeneous Model, the solution to \eqref{eq.67} is
  \begin{equation}\label{eq.16}
r_l^{\rm Dep}:=\frac{\E{\left( \Lambda^{[1]}\Lambda^{[2]}\right)^2  \Theta^{[1]}\Theta^{[2]}\bigg\vert L=l}}{\E{\left(\Lambda^{[1]} \Lambda^{[2]}\right)^2\big\vert L=l}}
\quad\hbox{for}\quad l=1, \cdots,{z}
\end{equation}
where the numerator and denominator can be calculated by
\[
\E{\left( \Lambda^{[1]}\Lambda^{[2]}\right)^2 \Theta^{[1]}\Theta^{[2]}\bigg\vert L=l}=\frac{\sum\limits_{\kappa\in\mathcal{K}} w_{\kappa} \left(\lambda_{\kappa}^{[1]}\lambda_{\kappa}^{[2]}\right)^2  \int\int \theta^{[1]}\theta^{[2]}  \pi_l(\lambda_{\kappa}^{[1]}\theta^{[1]}) h(\theta^{[1]}, \theta^{[2]}){\rm d}\theta^{[1]}{\rm d} \theta^{[2]}}
{\sum\limits_{\kappa\in\mathcal{K}}w_{\kappa} \int  \pi_l(\lambda_{\kappa}^{[1]}\theta^{[1]}) g_1(\theta^{[1]}){\rm d}\theta^{[1]}  }
\]
and
\[
\E{\left( \Lambda^{[1]}\Lambda^{[2]}\right)^2\big\vert L=l} = \frac{\sum\limits_{\kappa\in\mathcal{K}} w_{\kappa} \left(\lambda_{\kappa}^{[1]}\lambda_{\kappa}^{[2]}\right)^2  \int \pi_l(\lambda_{\kappa}^{[1]}\theta^{[1]})g_1(\theta^{[1]}){\rm d}\theta^{[1]}}
{\sum\limits_{\kappa\in\mathcal{K}} w_{\kappa}  \int \pi_l(\lambda_{\kappa}^{[1]}\theta^{[1]}) g_1(\theta^{[1]}){\rm d}\theta^{[1]} },
\]
respectively.
\end{theorem}

%

The following corollary shows
the optimal relativity under the Heterogeneous Model when the frequency and severity are assumed to be independent.
\begin{corollary}\label{cor.1}
  Under the Heterogeneous Model, when $\Theta^{[1]}$ and $\Theta^{[2]}$ are independent\footnote{The copula in \eqref{eq.4} is an independent copula}, the solution to \eqref{eq.67} is \begin{equation}\label{eq.1600}
r_l^{\rm Indep}:=\frac{\E{\left( \Lambda^{[1]}\Lambda^{[2]}\right)^2  \Theta^{[1]}\bigg\vert L=l}}{\E{\left(\Lambda^{[1]} \Lambda^{[2]}\right)^2\big\vert L=l}}
\quad\hbox{for}\quad l=1, \cdots,{z},
\end{equation}
where the numerator and denominator can be calculated by
\[
\E{\left( \Lambda^{[1]}\Lambda^{[2]}\right)^2 \Theta^{[1]}\bigg\vert L=l}=\frac{\sum\limits_{\kappa\in\mathcal{K}} w_{\kappa} \left(\lambda_{\kappa}^{[1]}\lambda_{\kappa}^{[2]}\right)^2  \int \theta^{[1]}  \pi_l(\lambda_{\kappa}^{[1]}\theta^{[1]}) g_1(\theta^{[1]}){\rm d}\theta^{[1]}}
{\sum\limits_{\kappa\in\mathcal{K}}w_{\kappa} \int  \pi_l(\lambda_{\kappa}^{[1]}\theta^{[1]}) g_1(\theta^{[1]}){\rm d}\theta^{[1]}  }
\]
and
\[
\E{\left( \Lambda^{[1]}\Lambda^{[2]}\right)^2\big\vert L=l} = \frac{\sum\limits_{\kappa\in\mathcal{K}} w_{\kappa} \left(\lambda_{\kappa}^{[1]}\lambda_{\kappa}^{[2]}\right)^2  \int \pi_l(\lambda_{\kappa}^{[1]}\theta^{[1]})g_1(\theta^{[1]}){\rm d}\theta^{[1]}}
{\sum\limits_{\kappa\in\mathcal{K}} w_{\kappa}  \int \pi_l(\lambda_{\kappa}^{[1]}\theta^{[1]}) g_1(\theta^{[1]}){\rm d}\theta^{[1]} },
\]
respectively.
\end{corollary}

Apart from peripheral differences, the classical BMS literature defines the optimal relativity $r_l$s as the solution to the optimization problem for the frequency only.
For example, \citet{Chong} defines the optimal relativity as
\begin{equation*}
\begin{aligned}
(r_1^{\rm Tan}, \cdots, r_{{z}}^{\rm Tan})&:=\argmin_{(r_1, \cdots, r_{{z}})\in \Real^{{z}}} \E{\left(\E{N_{t+1}\big\vert \Lambda^{[1]}, \Theta^{[1]}}-r_L\Lambda^{[1]}\right)^2}\\
&=\argmin_{(r_1, \cdots, r_{{z}})\in \Real^{{z}}} \E{\left(\Lambda^{[1]}\Theta^{[1]}-\Lambda^{[1]} r_L\right)^2},\\
\end{aligned}
\end{equation*}
which has the following analytical expression
\begin{equation}\label{eq.60}
r_l^{\rm Tan}=\frac{\E{\left(\Lambda^{[1]}\right)^2 \Theta^{[1]}\big\vert L=l}}{\E{\left(\Lambda^{[1]}\right)^2\big\vert L=l}}, \quad l=1, \cdots, {z}.
\end{equation}

\begin{remark}
When $\boldsymbol{X}^{[1]}$ and $\boldsymbol{X}^{[2]}$ are independent,
the optimal relativity in \eqref{eq.1600} is the same as the optimal relativity in \eqref{eq.60}.
However, if $\boldsymbol{X}^{[1]}$ and $\boldsymbol{X}^{[2]}$ are not independent, which is generally true if $\boldsymbol{x}^{[1]}$ and $\boldsymbol{x}^{[2]}$ share a common element, then
the optimal relativity in \eqref{eq.1600} is not the same as the optimal relativity in \eqref{eq.60}.
\end{remark}

\section{Data Analysis}\label{sec.6}

To examine the effect of dependence in BM scales, we use a data set for collision coverage for new and old vehicles from the Wisconsin Local Government Property Insurance Fund (LGPIF) in \citet{Frees4}. Detailed information on the project is available on the LGPIF project website \citep{LGPIF}.
 The LGPIF provides
property insurance\footnote{LGPIF coverage is categorized into three major types: building and contents, inland marine, and motor vehicles.} for various governmental entities, including counties, cities, towns, villages, school
districts, fire departments, and other miscellaneous entities.
Collision coverage provides coverage for the impact of a vehicle with an object, impact of vehicle with an attached vehicle, or the overturn of a vehicle.
The longitudinal data from 1,234 local government entities covers from 2006 to 2010. We removed the observations for policyholders whose new collision coverage and old collision coverage are zero. Hence, we use longitudinal data from 497 governmental entities in our data analysis.
We have two categorical variables: the entity type with six levels, miscellaneous, city, county, school, town and village, and average, and coverage with 3 levels, coverage 1 $\in (0, 0.14]$, coverage 2 $\in (0.14,0.74]$, and coverage 3 $\in (0.74,\infty]$.
Table \ref{tab.x} describes the distribution of the types of local government entities.
Furthermore, Table \ref{tab.n} summarizes the claim frequency for the policy years of 2006 to 2010. This table suggests that the claims distribution appears to be stable over time.
Tables \ref{entity_year} and \ref{coverage_year} provide average frequency and severity by entity type for 2006 to 2010, and by coverage and years, respectively.



In the following subsection, we specify the Heterogeneous Model for the data analysis.

\subsection{Model Specification}\label{sec.6.1}

We use the entity type and coverage, which are the categorical variables, as explanatory variables in $\lambda_i^{[1]}$ and $\lambda_i^{[2]}$. Specifically, with log link functions, we have
\[
\lambda_i^{[1]}=\exp(\boldsymbol{x}_{i}^{[1]}\boldsymbol{\beta}^{[1]}), \quad\hbox{and}\quad
\lambda_i^{[2]}=\exp(\boldsymbol{x}_{i}^{[2]}\boldsymbol{\beta}^{[2]}),
\]
where $\boldsymbol{x}_i^{[1]}$ and $\boldsymbol{x}_i^{[2]}$ are the model matrices\footnote{In this example, we set $\boldsymbol{x}_i^{[1]}=\boldsymbol{x}_i^{[2]}$.} of two categorical variables for frequency and severity, and $\boldsymbol{\beta}^{[1]}$ and $\boldsymbol{\beta}^{[2]}$ are the corresponding parameters.
Assuming the frequency-severity dependence as in the Heterogeneous Model with
\begin{equation}\label{eq.pri.1}
\left(\Theta_i^{[1]},\Theta_i^{[2]}\right) \sim H=C_\rho(G_1, G_2), \quad i=1, \cdots, I,
\end{equation}
where we assume $C_\rho$ is the bivariate Gaussian copula with correlation parameter $\rho$, and marginals are given by
\[
\Theta_i^{[1]}\sim {\rm Log}{\rm Normal}\left(-\sigma_1^2/2, \sigma_1^2\right)\quad\hbox{and}\quad
\Theta_i^{[2]}\sim {\rm Log Normal}\left(-\sigma_2^2/2, \sigma_2^2\right).
\]

%

For the prior distributions of the parameters, we assume independence among the parameters
($\boldsymbol{\beta}^{[1]}$, $\boldsymbol{\beta}^{[2]}$, $\psi^{[2]}$, $\sigma_1$, $\sigma_2$, $\rho$). We use a diffuse normal prior for $\boldsymbol{\beta}^{[1]}$ and $\boldsymbol{\beta}^{[2]} $; that is,
$\boldsymbol{\beta}^{[1]} \sim {\rm MVN}(\boldsymbol{a}_0^{[1]},\boldsymbol{A}_0^{[1]})$, and
$\boldsymbol{\beta}^{[2]} \sim {\rm MVN}(\boldsymbol{a}_0^{[2]},\boldsymbol{A}_0^{[2]})$,
where
$\boldsymbol{a}_0^{[1]}$ and $\boldsymbol{a}_0^{[2]}\in\Real^8$ are the mean vectors and
$\boldsymbol{A}_0^{[1]}$ and $\boldsymbol{A}_0^{[2]}\in \Real^8\times \Real^8$ are the covariance matrices.
We choose
$1/\psi^{[2]} \sim{\rm Gamma}(c_0,d_0)$,
$\sigma_1 \sim {\rm Uniform}(e_0,f_0)$,
$\sigma_2 \sim {\rm Uniform}(g_0,h_0)$, and
$\rho \sim {\rm Uniform}(-1,1)$,
where
$(c_0, d_0,e_0,f_0,g_0,h_0)$
are pre-specified hyper parameters.

From the Heterogeneous Model with the specified bivariate distribution in \eqref{eq.pri.1} and priors, the joint posterior distribution of the parameters
$(\boldsymbol{\beta}^{[1]},\boldsymbol{\beta}^{[2]},\psi^{[2]},\sigma_1,\sigma_2,\rho)$ is
\begin{equation*}
\begin{aligned}
\left\{ L_0 L_1\right\}& \prod_{i=1}^I c_\rho\left(G_1(\Theta_i^{[1]}), G_2(\Theta_i^{[2]})\right)g_1\left(\Theta_i^{[1]}\right)g_2\left(\Theta_i^{[2]}\right)  \\
& \cdot \pi(\boldsymbol{\beta}^{[1]}|\boldsymbol{a}_0^{[1]},\boldsymbol{A}_0^{[1]}) \pi(\boldsymbol{\beta}^{[2]}|\boldsymbol{a}_0^{[2]},\boldsymbol{A}_0^{[2]}) \pi(1/\psi^{[2]}|c_0,d_0) \pi(\sigma_1|e_0,f_0) \pi(\sigma_2|g_0,h_0) \pi(\rho),
\end{aligned}
\end{equation*}
where $L_0$ is the likelihood function for $N_{it}$
 having zero frequencies, while
 $L_1$ is the likelihood function for $N_{it}$ and $S_{it}$ having non-zero frequencies; specifically,
\begin{equation*}
L_0=\prod_{\{(i,t):N_{it}=0\}} f_1\left( N_{it}; \lambda_i^{[1]}\Theta_i^{[1]}\right),
\end{equation*}
and
\begin{equation*}
L_1= \prod_{\{(i,t):N_{it}>0\}} f_1\left(N_{it}; \lambda_i^{[1]}\Theta_i^{[1]} \right) f_2\left(S_{it}; N_{it}\lambda_i^{[2]}\Theta_i^{[2]}, \psi^{[2]}/N_{it} \right).
\end{equation*}
 Here, $c_\rho(u_1, u_2)$ is the density of the Gaussian copula with the correlation parameter $\rho$. 

For convenience, we call the Heterogeneous Model specified in this section the {\it dependent Heterogeneous Model} for $\rho\neq 0$. On the other hand, the model with the restriction of $\rho=0$ is the {\it independent Heterogeneous Model}.

\subsection{Estimation results based on the Heterogeneous Model}\label{sec.6.2}

Table \ref{jags_result} presents the summary statistics of the sample from the posterior distribution for each parameter with the Bayesian approach by using JAGS \citep{JAGS} to estimate the parameters of the dependent and independent Heterogeneous Models. We use 60,000 MCMC iterations and save every 5th sample after a burn-in of 20,000 iterations. The standard MCMC diagnostics gave no indication of a lack of convergence. We display the following three measures in the table: posterior median (median), posterior standard error (std.error), and 95$\%$ credible interval (95$\%$ CI). For a Bayesian credible interval, we use the 95 $\%$ highest posterior density (HPD) credible interval.
Note that * indicates that the parameters are significant at the 0.05 level.
Moreover, with the posterior samples, the deviance information criteria (DIC, \citet{Spiegelhalter2002}) for dependent and independent Heterogeneous Models are given in the bottom row of Table \ref{jags_result}.
Between the two models, the dependent Heterogeneous Model is preferred as it has the smallest DIC.

\begin{table}[h!]
\caption{Summary statistics of 60,000 posterior samples for the model parameters }
\centering
\begin{tabular}{ l r r r r l r r r r l c c }
 \hline
 &\multicolumn{5}{l}{Dependent}&\multicolumn{5}{l}{Independent}\\ \cline{2-5} \cline{7-10}

 &&&\multicolumn{2}{c}{95$\%$ CI}& &&&\multicolumn{2}{c}{95$\%$ CI}\\
parameter& median & std.error & lower & upper& & median & std.error & lower & upper \\
 \hline
 \multicolumn{11}{l}{ Frequency} \\
\quad Intercept &	-2.767& 	0.318&  -3.417&  -2.153&*& -2.761& 	0.319&  -3.393&  -2.142&*\\
\quad City      & 	 0.597& 	0.337&  -0.051& 	1.272& &  0.574& 	0.349&  -0.111& 	1.251\\
\quad County    &  1.907& 	0.335& 	1.271& 	2.587&*&	 1.877& 	0.345& 	1.206& 	2.551&*\\
\quad School    &	 0.411& 	0.304&  -0.181& 	1.014& &	 0.405& 	0.310&  -0.210& 	1.009\\
\quad Town      &	-1.351& 	0.384&  -2.103&  -0.584&*& -1.386& 	0.384&  -2.150&  -0.638&*\\
\quad Village   &	-0.012& 	0.323&  -0.626& 	0.654& & -0.034& 	0.332&  -0.716& 	0.599\\
\quad Coverage2 &	 1.247& 	0.212& 	0.829& 	1.667&*&  1.245& 	0.215& 	0.835& 	1.665&*\\
\quad Coverage3 &  2.139& 	0.230& 	1.713& 	2.615&*& 	 2.154& 	0.231& 	1.708& 	2.616&*\\
\hline
 \multicolumn{11}{l}{ Severity} \\

\quad Intercept &  	 8.829& 	0.375& 	8.103& 	9.588&*& 	 8.633& 	0.368& 	7.903& 	9.335&*\\
\quad City   	  & 		-0.036& 	0.353&  -0.737& 	0.637& & -0.053& 	0.345&  -0.728& 	0.631\\
\quad County 	  & 		 0.341& 	0.338&  -0.336& 	0.980& & 	 0.394& 	0.338&  -0.266& 	1.058\\
\quad School 	  & 		-0.173& 	0.328&  -0.805& 	0.484& & -0.155& 	0.326&  -0.792& 	0.487\\
\quad Town  	  &		 0.497& 	0.440&  -0.356& 	1.349& &  0.438& 	0.447&  -0.441& 	1.310\\
\quad Village   &		 0.316& 	0.346&  -0.357& 	0.994& & 	 0.289& 	0.348&  -0.406& 	0.951\\
\quad Coverage2 &		 0.180& 	0.244&  -0.308& 	0.646& & 	 0.159& 	0.243&  -0.308& 	0.640\\
\quad Coverage3 & 		-0.027& 	0.261&  -0.533& 	0.493& &  0.022& 	0.258&  -0.468& 	0.533\\
\quad $1/\psi^{[2]}$&   0.670& 	0.041& 	0.592& 	0.752&*& 	 0.678& 	0.041& 	0.600& 	0.761&*\\
\hline
\multicolumn{11}{l}{ Copula part} \\

\quad $\sigma^2_1$&  	 0.992& 	0.142& 	0.746& 	1.292&*& 	0.994& 	0.140& 	0.741& 	1.283&*\\
\quad $\sigma^2_2$&  	 0.293& 	0.067& 	0.176& 	0.433&*& 	0.297& 	0.065& 	0.181& 	0.432&*\\
\quad $\rho$&       	-0.447& 	0.130&  -0.690&  -0.190&*\\
 \hline
 DIC & \multicolumn{5}{c}{16530.76} & \multicolumn{5}{c}{16618.01} \\
 \hline
\end{tabular}
\label{jags_result}
\end{table}
The estimation results from the comparison of the DICs of the two models and the $95\%$ HPD interval indicates that
the estimate for the correlation coefficient between two random effects $\rho=-0.449$ is significantly negative.

Figure \ref{totalS} presents the posterior predictive distributions of the total severities for next year obtained from the posterior samples under the dependence model (solid line) and the independence model (dashed line).
Specifically, the total severity is the sum of all severities over policyholders, ${\rm Total}~S_{t+1} = \sum_{i=1}^{I} S_{i,t+1}$, indicating the amount of severity that the insurance company could be interested in.
Here, the dotted line indicates the total aggregate severity for the 2011 held data with 1,098 observations from 379 local government entities.
As Figure \ref{totalS} shows, the negative dependence makes the severities more left-skewed.


\begin{figure}[h!]
\centering
\includegraphics[width=\textwidth, height=2.2in]{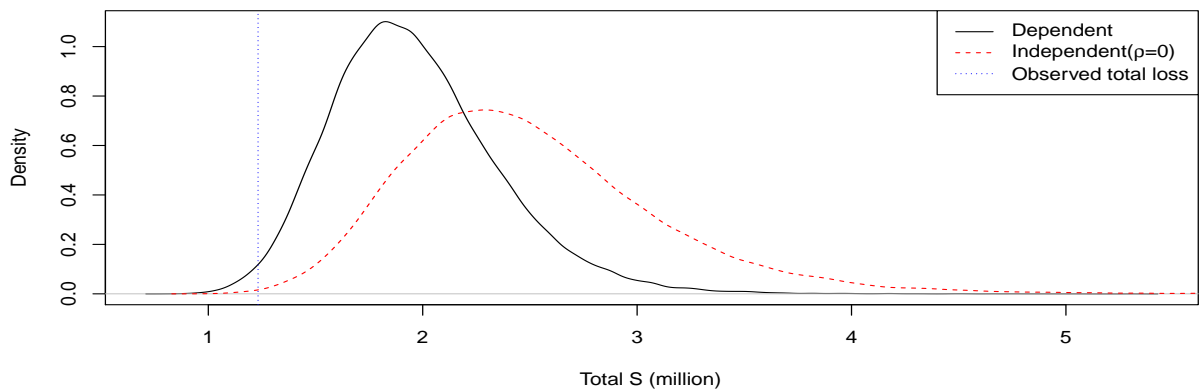}
\caption{ Posterior predictive distribution of total sum of severities   }
\label{totalS}
\end{figure}

Moreover, based on the estimation results, we can examine
the dependence between frequency and severity, between frequencies, and between severities:
 \begin{equation*}
{\rm corr}\left(N_{t},S_{t}\big\vert \boldsymbol{x}^{[1]}, \boldsymbol{x}^{[2]}\right), \,
{\rm corr}\left(N_{t_1},N_{t_2}\big\vert \boldsymbol{x}^{[1]}, \boldsymbol{x}^{[2]}\right)
,\, \hbox{and} \,\,\,{\rm corr}\left(S_{t_1},S_{t_2}\big\vert \boldsymbol{x}^{[1]}, \boldsymbol{x}^{[2]}\right)
\end{equation*}
which we calculate in Example \ref{ex.prop.1}.
Figure \ref{corrNS_corrSS} shows the posterior predictive distribution of the correlation coefficients for each group.
There are 18 groups depending on the a priori characteristics, entity type, and coverage level.
Each entity type has three boxplots (white, light gray, gray), which represent coverage 1, coverage 2, and coverage 3, respectively.

In Figure \ref{corrNS_corrSS} (a), the red boxplots represent the distribution of the correlation between frequency and aggregate severity under the independent Heterogeneous Model, which we can use as a reference to assess the effect of dependence in the dependent Heterogeneous Model.
In terms of the correlation between frequency and aggregate severity, we cannot use a correlation coefficient of zero, $\rho=0$, as a reference due to the inherited dependence between frequency and aggregate severity, as we show in Example \ref{ex.prop.1}.

The box plots of all groups in Figure \ref{corrNS_corrSS} (a) indicate significant negative dependence between frequency and aggregate severity.
Figure \ref{corrNS_corrSS} (b) and (c) show the distribution of the correlation between frequencies, and the distribution between severities, respectively.
Compared with the frequency-frequency dependence, severity-severity dependence is weaker on average.
Moreover, among the three types of dependence in Figure \ref{corrNS_corrSS}, the discrepancy in the distribution of correlation depending on the risk groups tend to be smaller in frequency-severity dependence.

\begin{figure}[p]
    \centering
        \subfloat[]{\includegraphics[width=\textwidth, height=2.4in]{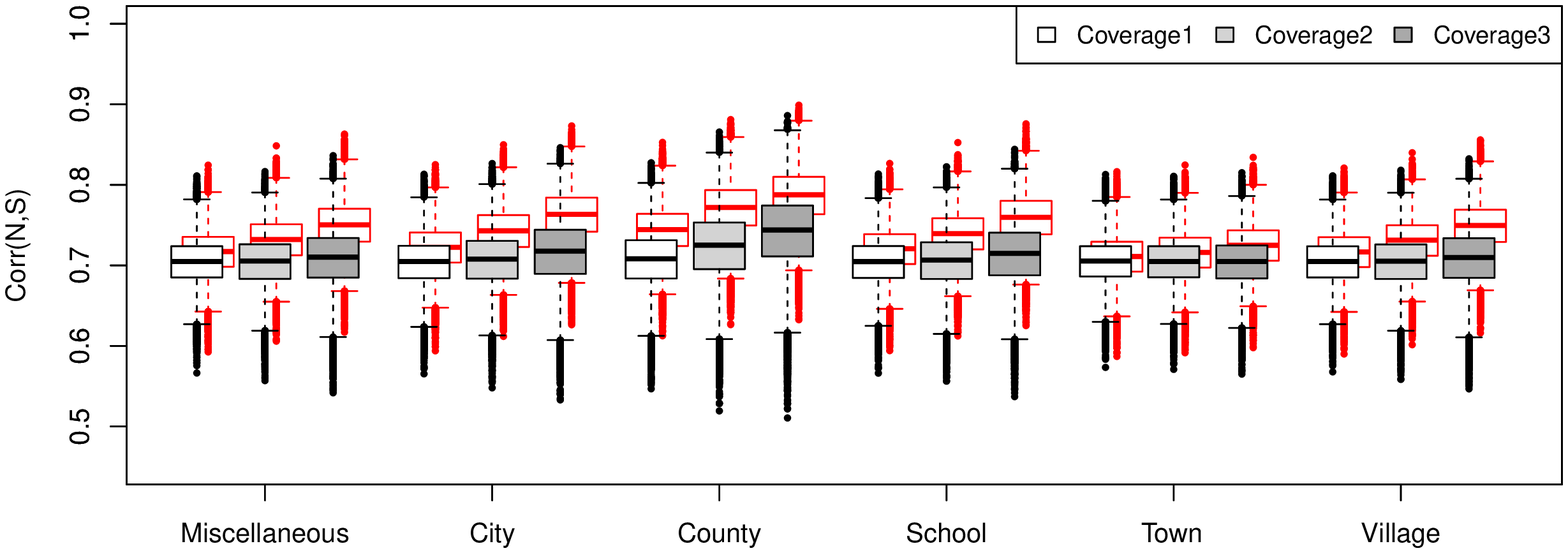}}

        \subfloat[]{\includegraphics[width=\textwidth, height=2.4in]{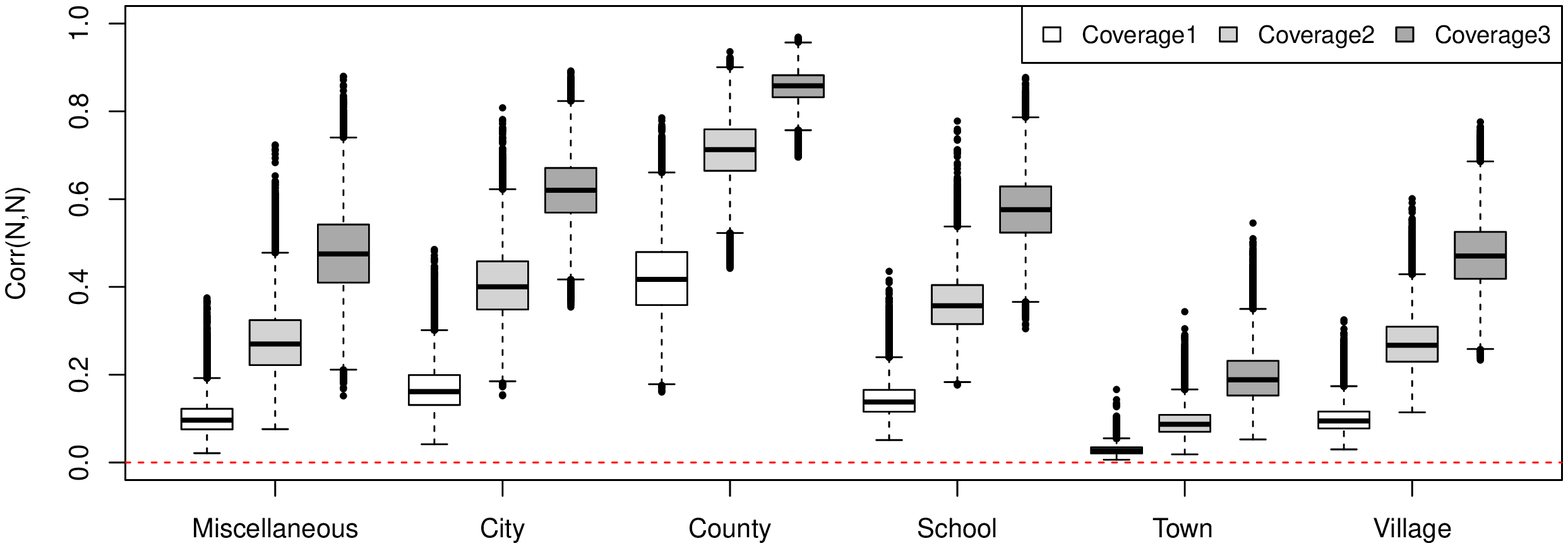}}

        \subfloat[]{\includegraphics[width=\textwidth, height=2.4in]{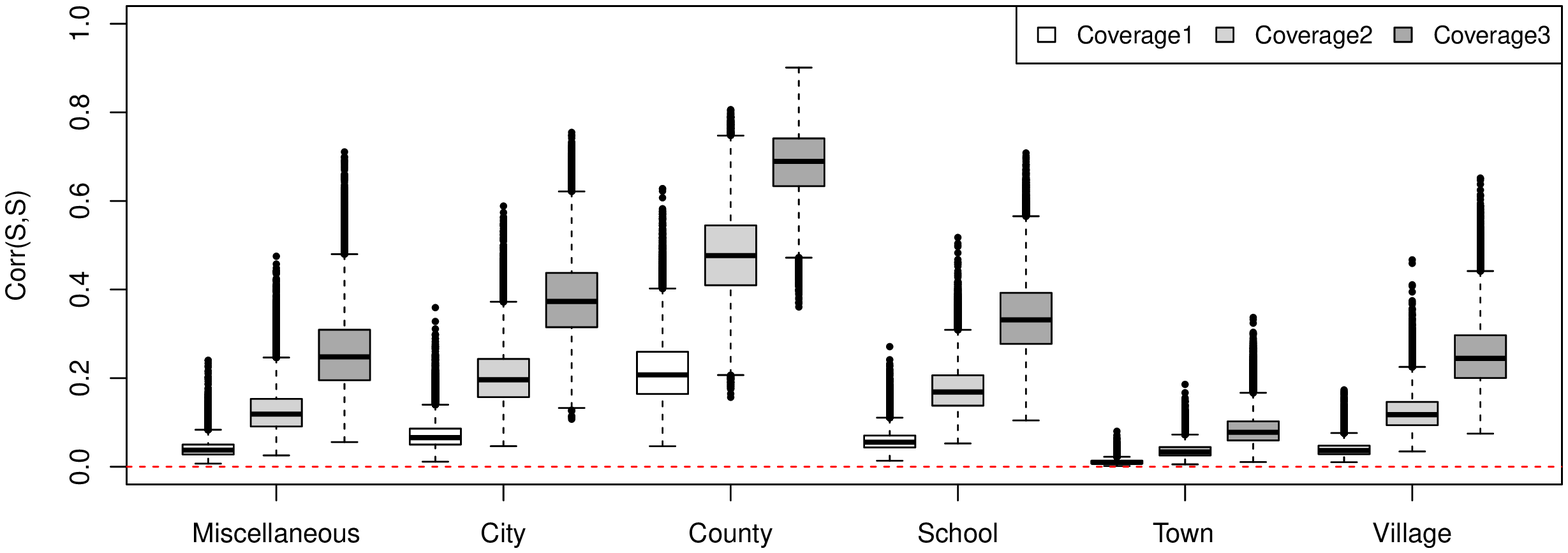}}

\caption{ Box plot of correlations (a) between frequency and aggregate severity, (b) between frequencies, and (c) between aggregate severities from the posterior samples. Red represents the criteria of independence.}
\label{corrNS_corrSS}
\end{figure}

\begin{landscape}
\begin{table}
\caption{Distribution of $L$ and various BM relativities under various transition rules}
\centering
\begin{tabular}{ l c c c g c c c c g c c c c g}
 \hline
 	& \multicolumn{4}{l}{-1/+1 system} & & \multicolumn{4}{l}{-1/+2 system}& & \multicolumn{4}{l}{-1/+3 system} \\ \cline{2-5} \cline{7-10}\cline{12-15}
	& \multicolumn{3}{l}{Relativity} &  \multirow{1}{*}{$P(L=l)$}  & & \multicolumn{3}{l}{Relativity} & \multirow{1}{*}{$P(L=l)$}& & \multicolumn{3}{l}{Relativity} & \multirow{1}{*}{$P(L=l)$}  \\
 Level $l$ 	& $r_l^{\rm Dep}$ & $r_l^{\rm Indep}$ & $r_l^{\rm Tan}$ & & & $r_l^{\rm Dep}$ & $r_l^{\rm Indep}$ & $r_l^{\rm Tan}$ & & & $r_l^{\rm Dep}$ & $r_l^{\rm Indep}$ & $r_l^{\rm Tan}$& \\
 \hline
10&	0.966&	1.278&	1.296&	0.145&&	0.932&	1.225&	1.241&	0.178	&&	0.920	&	1.207	&	1.221&	0.198\\
9&	0.454&	0.446&	0.483&	0.034&&	0.428&	0.415&	0.452&	0.051	&&	0.419	&	0.405	&	0.441&	0.064\\
8&	0.369&	0.337&	0.380&	0.018&&	0.343&	0.309&	0.352&	0.030	&&	0.336	&	0.301	&	0.343&	0.039\\
7&	0.327&	0.288&	0.335&	0.013&&	0.302&	0.260&	0.306&	0.023	&&	0.295	&	0.253	&	0.298&	0.033\\
6&	0.302&	0.259&	0.308&	0.012&&	0.278&	0.234&	0.282&	0.020	&&	0.273	&	0.228	&	0.274&	0.028\\
5&	0.284&	0.239&	0.290&	0.013&&	0.262&	0.216&	0.263&	0.024	&&	0.260	&	0.214	&	0.259&	0.025\\
4&	0.271&	0.224&	0.275&	0.017&&	0.253&	0.206&	0.253&	0.024	&&	0.249	&	0.201	&	0.242&	0.048\\
3&	0.259&	0.212&	0.262&	0.029&&	0.245&	0.197&	0.240&	0.055	&&	0.245	&	0.197	&	0.236&	0.041\\
2&	0.250&	0.203&	0.249&	0.074&&	0.243&	0.195&	0.236&	0.046	&&	0.243	&	0.194	&	0.232&	0.035\\
1&	0.242&	0.194&	0.231&	0.646&&	0.238&	0.189&	0.217&	0.548	&&	0.237	&	0.189	&	0.211&	0.489\\
\hline\hline																			
HMSE	&	91.24&	101.66&\multicolumn{2}{l}{101.74}&&	93.48&	102.25&\multicolumn{2}{l}{103.18}&& 94.28& 102.84&\multicolumn{2}{l}{103.68 } \\ \hline
\end{tabular}
\label{bm}
\end{table}

\begin{table}
\caption{Various BM relativities under various degrees of dependence}
\centering
\begin{tabular}{ l c c c c c c c c c c c c c }
 \hline
 	& \multicolumn{3}{l}{$\rho=$-0.8} & & \multicolumn{3}{l}{$\rho=$0}& & \multicolumn{3}{l}{$\rho=$0.4} \\ \cline{2-4} \cline{6-8}\cline{10-12}
	& \multicolumn{3}{l}{Relativity}  & & \multicolumn{3}{l}{Relativity} &  & \multicolumn{3}{l}{Relativity}  \\
 Level $l$ 	& $r_l^{\rm Dep}$ & $r_l^{\rm Indep}$ & $r_l^{\rm Tan}$ && $r_l^{\rm Dep}$ & $r_l^{\rm Indep}$ & $r_l^{\rm Tan}$ && $r_l^{\rm Dep}$ & $r_l^{\rm Indep}$ & $r_l^{\rm Tan}$ \\
 \hline
10&	0.769	&	1.278	&	1.296	&&	1.277	&	1.278	&	1.296	&&	1.624	&	1.278	&	1.296	\\
9&	0.447	&	0.446	&	0.483	&&	0.449	&	0.446	&	0.483	&&	0.432	&	0.446	&	0.483	\\
8&	0.382	&	0.337	&	0.380	&&	0.342	&	0.337	&	0.380	&&	0.312	&	0.337	&	0.380	\\
7&	0.349	&	0.288	&	0.335	&&	0.294	&	0.288	&	0.335	&&	0.263	&	0.288	&	0.335	\\
6&	0.328	&	0.259	&	0.308	&&	0.267	&	0.259	&	0.308	&&	0.236	&	0.259	&	0.308	\\
5&	0.312	&	0.239	&	0.290	&&	0.248	&	0.239	&	0.290	&&	0.218	&	0.239	&	0.290	\\
4&	0.300	&	0.224	&	0.275	&&	0.235	&	0.224	&	0.275	&&	0.206	&	0.224	&	0.275	\\
3&	0.289	&	0.212	&	0.262	&&	0.224	&	0.212	&	0.262	&&	0.197	&	0.212	&	0.262	\\
2&	0.280	&	0.203	&	0.249	&&	0.217	&	0.203	&	0.249	&&	0.190	&	0.203	&	0.249	\\
1&	0.272	&	0.194	&	0.231	&&	0.210	&	0.194	&	0.231    &&	0.184	&	0.194	&	0.231	\\
\hline\hline																			
HMSE	&24.92	&    49.98	&   51.68     && 325.42  	& 	325.42 	& 	325.50 	&&	887.56	&	899.06	& 	898.02  \\ \hline
\end{tabular}
\label{bm_dep}
\end{table}

\end{landscape}

\subsubsection{Further Analysis with the Individual Severity Assumption}

As Figure \ref{corrNS_corrSS} (a) shows, the correlation between frequency and aggregate severity is not a convenient statistic to investigate the frequency-severity dependence structure. For a more intuitive comparison of the frequency-severity dependence structure, it is better to obtain the correlation between frequency and individual severities.
Based on the estimation results in Section \ref{sec.6.2} and Assumption \ref{ahn.assumption}, the following statistics can be implied
 \begin{equation*}
{\rm corr}\left(N_{t},Y_{t,j}\bigg\vert \boldsymbol{x}^{[1]}, \boldsymbol{x}^{[2]}\right)
,\quad {\rm and } \quad{\rm corr}\left(Y_{t_1,j_1},Y_{t_2,j_2}\bigg\vert \boldsymbol{x}^{[1]}, \boldsymbol{x}^{[2]}\right)
\end{equation*}
as in Example \ref{ex.prop.2}.

Figure \ref{corrNY_corrYY} (a) shows the distribution of the correlations between frequency and severity.
The box plots of all groups in the figure tend to have significant negative dependence between frequency and individual severity.
Figure \ref{corrNY_corrYY} (b) presents the distribution of the correlations among individual severities.
As we show in Example \ref{ex.prop.2}, since the correlation between severities does not depend on the risk groups, we can represent it in a single boxplot.

\begin{figure}[!ht]
    \centering
        \subfloat[]{\includegraphics[width=0.8\textwidth, height=2.5in]{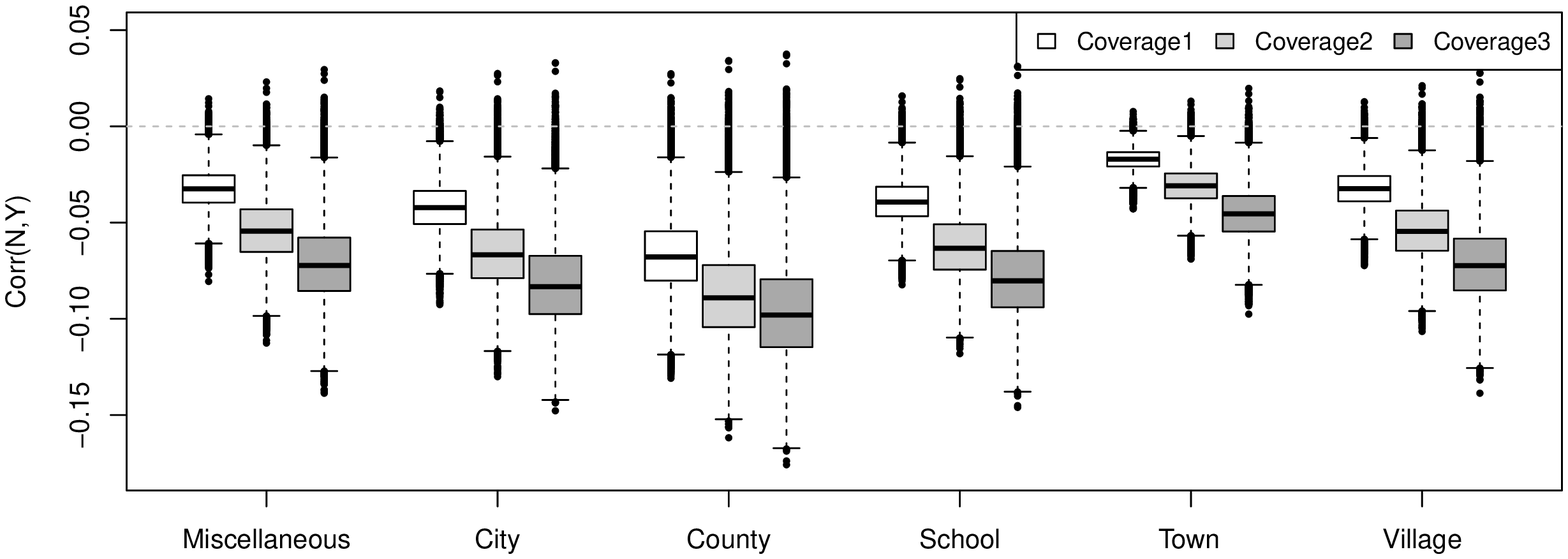}}
        \qquad
        \subfloat[]{\includegraphics[width=0.15\textwidth, height=2.5in]{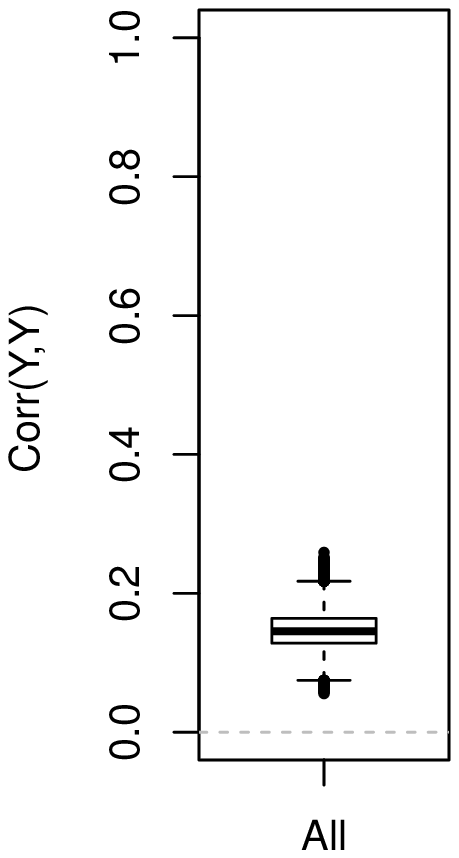}}
        \caption{ Box plot of correlations: (a) between frequency and severity and (b) between severities from the posterior samples. The horizontal dashed lines represent the criteria of independence. }
    \label{corrNY_corrYY}
\end{figure}

\subsection{BM Relativities}\label{sec.6.3}
From the estimation results\footnote{We employ the median of the posterior sample with the model in section \ref{sec.6.1} to calculate relativities.}
 under the dependent and the independent Heterogeneous Models, Table \ref{bm} reports the BM relativities defined in \eqref{eq.16}, \eqref{eq.1600}, and \eqref{eq.60} with $z=10$ BM levels under the -1/+1 and -1/+2 transition rules.
 Under both systems, except relativity at level 10, the relativities under the dependence model are slightly larger than those under the model that assumes independence between frequency and severity, which coincides with our intuition of negative dependence between frequency and severity.

To compare the quality of the various optimal relativities,
we propose a measure called the {\it hypothetical mean square error} (HMSE),
 which is defined as the mean square error (MSE) between the aggregate severity and the premium determined by BMS as defined in \eqref{eq.67}.
%



\begin{proposition}\label{prop3}
Under the Heterogeneous Model with given relativity $\boldsymbol{r}$, the MSE in \eqref{eq.67} is
 \begin{equation}\label{eq.68}
 \begin{aligned}
 {\rm HMSE}(\boldsymbol{r})
&:=\E{\left(S_{t+1}-\Lambda^{[1]} \Lambda^{[2]}  \boldsymbol{r}_L \right)^2}\\
&=\sum\limits_{\kappa\in\mathcal{K}} w_\kappa
\int\int \sum\limits_{l=1}^{z} \left(\lambda_\kappa^{[1]}\lambda_\kappa^{[2]} \theta^{[1]}\theta^{[2]}-
\lambda_\kappa^{[1]}\lambda_\kappa^{[2]}r_l
\right)^2
\pi_l\left(\lambda_\kappa^{[1]}\theta^{[1]}\right) h(\theta^{[1]}, \theta^{[2]}) {\rm d}\theta^{[1]}{\rm d}\theta^{[2]}.
\end{aligned}
 \end{equation}
\end{proposition}

Under the Heterogeneous Model, ${\rm HMSE}\left(\boldsymbol{r}^{\rm Dep}\right)$ with $\boldsymbol{r}^{\rm Dep}$ defined as in \eqref{eq.16} is guaranteed to be
smaller than ${\rm HMSE}(\boldsymbol{r})$ with any other choice of $\boldsymbol{r}$ by definition.
Table \ref{bm} shows the significant improvements in HMSE under the proposed relativities $\boldsymbol{r}^{\rm Dep}$. This improvement increases with stronger dependence, as Table \ref{bm_dep} shows.

\section{Conclusion}\label{sec.7}
In modeling aggregate severity, it is customary to assume independence between the claim frequency and severity.
Although, this independence model helps to simplify ongoing ratemaking systems such as BMS, it is often unwarranted.
 This paper proposes a flexible method to define dependence between the claim frequency and severity using a copula-based bivariate random effects model.
 A copula representation of bivariate random effects enables a convenient extension of marginal models that are already available in the literature.

Based on the proposed random effect model,
 our data analysis using the proposed model shows the significant frequency-severity dependence.
 Additionally, we show how to incorporate the frequency-severity dependence into the frequency based BMS.
 Our numerical experiments show that the predictions with the proposed BMS outperform the predictions of the classical BMS.



%

\section*{Acknowledgements}
 Jae Youn Ahn was supported by a National Research Foundation of Korea (NRF) grant funded by the Korean Government (NRF-2017R1D1A1B03032318).

\bibliographystyle{apalike}
\bibliography{Bib_Oh}

\newpage
\appendix

\section{Implication for Individual Severities} \label{sec.3.1}

The Heterogeneous Model specifies the distribution assumption of frequency and aggregate severity $(N_t, S_t)$. However, knowing that aggregate severity is the sum of the individual severities $\left( Y_{t,1}, \cdots, Y_{t,N_t}\right)$, it is natural to think of the distributional assumption of individual severities. This section studies the distributional assumption of individual severities, specifically the following:

\bigskip
\begin{assumption}\label{ahn.assumption}
   For severity, we specify a conditional distribution of the individual severities
   conditional on the observed and unobserved risk characteristics
        \begin{equation}\label{a.eq.1}
        \left(\boldsymbol{X}^{[1]}, \boldsymbol{X}^{[2]}\right)=\left(\boldsymbol{x}^{[1]}, \boldsymbol{x}^{[2]}\right)
        \quad\hbox{and}\quad \left(\Theta^{[1]}, \Theta^{[2]}\right)
    =\left(\theta^{[1]}, \theta^{[2]}\right).
        \end{equation}
   Specifically, we assume $N_{t}$ and $\left( Y_{t,1}, \cdots, Y_{t,z} \right)$ are independent for any $z\in\mathbb{N}$ conditional on the observed and unobserved risk characteristics in \eqref{a.eq.1}, and the conditional distribution is
  \begin{equation}\label{eq.y}
  Y_{t,j}\big\vert \left(\boldsymbol{x}^{[1]},\boldsymbol{x}^{[2]}, \theta^{[1]},
        \theta^{[2]}\right)  \iid {F}_2\left(\cdot\,;\, \lambda^{[2]}\theta^{[2]}, \psi^{[2]}\right), \quad 
  t, j\in \mathbb{N},
  \end{equation}
   where the distribution $F_2$ in the EDF has mean parameter $\lambda^{[2]}\theta^{[2]}$ and dispersion parameter $\psi^{[2]}$.

\end{assumption}

Note that Assumption \ref{ahn.assumption} in the above model is well defined in the sense that the distribution of
\[
( Y_{t,1}, \cdots, Y_{t,z_1} )
\]
is implied by the distributional assumption of $( Y_{t,1}, \cdots, Y_{t,z_2} )$ for positive integers $z_1$ and $z_2$ satisfying $z_1<z_2$.
Assumption \ref{ahn.assumption} is hypothetical in that it includes the distributional assumption of the unobserved individual severity $y_{t,j}$ for $j>n_{t}$. This hypothetical assumption of the missing or unobserved data is common in the statistical literature \citep{Neuhaus2011}, and can be found in the recent frequency-severity literature on insurance \citep{Liu2017, Cossette2018}

An important observation is that the individual severity assumption in Assumption \ref{ahn.assumption} is coherent with the assumption of severity under the Heterogeneous Model.
Specifically, under the EDF assumption of the distribution $F_2$, the distributional assumption in \eqref{eq.y} implies \eqref{eq.s}.
Hence, the EDF assumptions of $F_2$ in the Heterogeneous Model and in Assumption \ref{ahn.assumption} are critical for the conditions in \eqref{eq.s} and \eqref{eq.y} to be coherent.

We further assume that for $F_1$ to be in the EDF,
the simple representations of the mean and variance of the aggregate severity are possible, as the following proposition shows.

%
%
%

  \begin{proposition}\label{prop.2}
 Under the Heterogeneous Model with Assumption \ref{ahn.assumption}, assume that the frequency distribution function ${F}_1$ is in ${\it EDF}$.
    Furthermore,
  let ${V}_1$ and ${V}_2$ be the variance functions for the frequency distribution $F_1$ and severity distribution $F_2$, respectively. Then, we have the following conditional expressions.

 \begin{enumerate}
 \item[i.] The variance of the individual severity conditional a priori are
\[
\Var{Y_{t,j}\bigg\vert \boldsymbol{x}^{[1]}, \boldsymbol{x}^{[2]}}
= \psi^{[2]}\E{{V}_2\left({\lambda^{[2]}} \Theta^{[2]}\right) \bigg\vert \boldsymbol{x}^{[1]}, \boldsymbol{x}^{[2]} } +\left(\lambda^{[2]}\right)^2\Var{\Theta^{[2]}}
\]
for $t, j\in\mathbb{N}$, and the covariance between individual severities is
\begin{equation}\label{eq.301}
\cov{Y_{t_1,j_1},Y_{t_2,j_2}\bigg\vert \boldsymbol{x}^{[1]}, \boldsymbol{x}^{[2]} }
=\left(\lambda^{[2]} \right)^2\Var{\Theta^{[2]}  }
\end{equation}
for $t_1, t_2, j_1, j_2 \in \mathbb{N}$ and $j_1\neq j_2$.

 \item[ii.] The covariance between the individual severity and frequency conditional a priori are
\[
\cov{N_{t_1}, Y_{t_2,j}\bigg\vert \boldsymbol{x}^{[1]}, \boldsymbol{x}^{[2]}} =
 \lambda^{[1]}\lambda^{[2]}\E{\Theta^{[1]}\Theta^{[2]}} -
\lambda^{[1]}\lambda^{[2]} 
\]
for $t_1,t_2, j \in \mathbb{N}$.

\end{enumerate}
\end{proposition}

\begin{proof}
  In the Heterogeneous Model, for the simplicity of calculation, we assume that the frequency distribution function ${F}_1$ is in ${\it EDF}(\lambda^{[1]}\Theta^{[1]}, \psi^{[1]})$.
  In addition to the assumption of the average severity in \eqref{eq.s}, we add an assumption of the individual severity in \eqref{eq.y}.

The variance of the individual severity conditional a priori are
\[
\begin{aligned}
&\Var{Y_{t,j}\bigg\vert \boldsymbol{x}^{[1]}, \boldsymbol{x}^{[2]}}\\
&= \E{ \Var{Y_{t,j} \big\vert  \boldsymbol{x}^{[1]}, \boldsymbol{x}^{[2]},\Theta^{[1]}, \Theta^{[2]}} \bigg\vert \boldsymbol{x}^{[1]}, \boldsymbol{x}^{[2]}}
+ \Var{ \E{Y_{t,j} \big\vert  \boldsymbol{x}^{[1]}, \boldsymbol{x}^{[2]},\Theta^{[1]}, \Theta^{[2]}} \bigg\vert \boldsymbol{x}^{[1]}, \boldsymbol{x}^{[2]}}\\
&= \psi^{[2]} \E{{V}_2\left({\lambda^{[2]}} \Theta^{[2]}\right) \bigg\vert \boldsymbol{x}^{[1]}, \boldsymbol{x}^{[2]} }
+ \left(\lambda^{[2]}  \right)^2\Var{\Theta^{[2]}}
\end{aligned}
\]
for $t, j=1,2,\cdots$, and the covariance between the individual severities is
\[
\begin{aligned}
& \cov{Y_{t_1,j_1},Y_{t_2,j_2}\bigg\vert \boldsymbol{x}^{[1]}, \boldsymbol{x}^{[2]}}\\
&=\E{ \cov{Y_{t_1,j_1}, Y_{t_2,j_2}\big\vert  \boldsymbol{x}^{[1]}, \boldsymbol{x}^{[2]},\Theta^{[1]}, \Theta^{[2]} } \bigg\vert \boldsymbol{x}^{[1]}, \boldsymbol{x}^{[2]}  } \\
& \quad\quad +\cov{ \E{Y_{t_1,j_1}\big\vert \boldsymbol{x}^{[1]}, \boldsymbol{x}^{[2]},\Theta^{[1]}, \Theta^{[2]} },
    \E{Y_{t_2,j_2}\big\vert \boldsymbol{x}^{[1]}, \boldsymbol{x}^{[2]},\Theta^{[1]}, \Theta^{[2]} }
     \bigg\vert \boldsymbol{x}^{[1]}, \boldsymbol{x}^{[2]} }\\
&=\left(\lambda^{[2]}  \right)^2 \Var{\Theta^{[2]} }
\end{aligned}
\]
for $t_1\neq t_2$ or $j_1\neq j_2$.
 The covariance between the individual severity and frequency conditional a priori are
\begin{equation}\label{eq.pf.1}
\begin{aligned}
&\cov{N_{t_1}, Y_{t_2,j}\bigg\vert \boldsymbol{x}^{[1]}, \boldsymbol{x}^{[2]}}\\
&=\E{ \E{N_{t_1}Y_{t_2,j}\big\vert  N_{t_1}, \boldsymbol{x}^{[1]}, \boldsymbol{x}^{[2]},\Theta^{[1]}, \Theta^{[2]}}\bigg\vert \boldsymbol{x}^{[1]}, \boldsymbol{x}^{[2]} } \\
& \quad\quad - \E{ \E{N_{t_1} \big\vert \boldsymbol{x}^{[1]}, \boldsymbol{x}^{[2]},\Theta^{[1]}, \Theta^{[2]}}\bigg\vert \boldsymbol{x}^{[1]}, \boldsymbol{x}^{[2]} }
\E{ \E{Y_{t_2,j}\big\vert \boldsymbol{x}^{[1]}, \boldsymbol{x}^{[2]},\Theta^{[1]}, \Theta^{[2]} }\bigg\vert \boldsymbol{x}^{[1]}, \boldsymbol{x}^{[2]} }\\
&= \E{ N_{t_1} \lambda^{[2]}\Theta^{[2]}  \bigg\vert \boldsymbol{x}^{[1]}, \boldsymbol{x}^{[2]} }
   - \lambda^{[1]}\lambda^{[2]}\E{\Theta^{[1]}}\E{\Theta^{[2]}} \\
&= \lambda^{[1]}\lambda^{[2]}\E{\Theta^{[1]}\Theta^{[2]}} - \lambda^{[1]}\lambda^{[2]}\E{\Theta^{[1]}}\E{\Theta^{[2]}}
\end{aligned}
\end{equation}
for $t_1\neq t_2$.

\end{proof}

Given the specific distribution, we can simplify the conditional expression in Proposition \ref{prop.2} further, as the following example shows.

\begin{example}\label{ex.prop.2}

 Under the Hetereogeneous Model with Assumption \ref{ahn.assumption},
consider the Poisson distribution for $F_1$ and Gamma distribution for $F_2$, with the log link functions for $\eta_1$ and $\eta_2$, respectively:
\[
     N_{t}\big\vert \left(\theta^{[1]}, \theta^{[2]}, \boldsymbol{x}^{[1]}, \boldsymbol{x}^{[2]}\right) \sim {\rm Pois}\left(\lambda^{[1]}\theta^{[1]}\right), \quad t\in \mathbb{N}
\]
and
\[
       Y_{t,j}\big\vert \left(\theta^{[1]}, \theta^{[2]}, \boldsymbol{x}^{[1]}, \boldsymbol{x}^{[2]}\right)  \sim {\rm Gamma}\left( \lambda^{[2]}\theta^{[2]}, \psi^{[2]}\right),\quad j=1, \cdots, z\quad\hbox{and}\quad t\in \mathbb{N},
\]
where
\[
\lambda^{[1]}=\exp\left(\boldsymbol{x}^{[1]}\boldsymbol{\beta}^{[1]}\right) \quad\hbox{and}\quad
\lambda^{[2]}=\exp\left(\boldsymbol{x}^{[2]} \boldsymbol{\beta}^{[2]}\right).
\]

Then, from Proposition \ref{prop.2}, we have the following covariance formulas:
\[
\begin{cases}
  {\rm cov}\left(N_{t_1},N_{t_2}\bigg\vert \boldsymbol{x}^{[1]}, \boldsymbol{x}^{[2]}\right)=\left(\lambda^{[1]}\right)^2\left(\exp(\sigma_1^2)-1\right);\\
  {\rm cov}\left(N_{t},Y_{t,j}\bigg\vert \boldsymbol{x}^{[1]}, \boldsymbol{x}^{[2]}\right)={\rm cov}\left(N_{t_1},Y_{t_2,j}\bigg\vert \boldsymbol{x}^{[1]}, \boldsymbol{x}^{[2]}\right)=\lambda^{[1]} \lambda^{[2]} \left(\exp(\rho\sigma_1\sigma_2)-1\right);\\
  {\rm cov}\left(Y_{t_1,j_1},Y_{t_2,j_2}\bigg\vert \boldsymbol{x}^{[1]}, \boldsymbol{x}^{[2]}\right)=\left(\lambda^{[2]}\right)^2 \left(\exp(\sigma_2^2)-1\right);
\end{cases}
\]
and the corresponding correlation formulas,
\[
\begin{cases}
{\rm corr}\left(N_{t_1},N_{t_2}\bigg\vert \boldsymbol{x}^{[1]}, \boldsymbol{x}^{[2]}\right)=\frac{\lambda^{[1]}\left\{\exp(\sigma_1^2)-1\right\}}{1+\lambda^{[1]}\left\{\exp(\sigma_1^2)-1\right\}};\\
{\rm corr}\left(N_{t},Y_{t,j}\bigg\vert \boldsymbol{x}^{[1]}, \boldsymbol{x}^{[2]}\right)=
{\rm corr}\left(N_{t_1},Y_{t_2,j}\bigg\vert \boldsymbol{x}^{[1]}, \boldsymbol{x}^{[2]}\right)=\frac{\lambda^{[1]} \left\{\exp(\rho\sigma_1\sigma_2)-1\right\}}{\sqrt{\lambda^{[1]}\left(1+\lambda^{[1]}\left\{\exp(\sigma_1^2)-1\right\}\right)}\sqrt{\left\{1+\psi^{[2]}\right\}\exp(\sigma_2^2)-1}};\\
{\rm corr}\left(Y_{t_1,j_1},Y_{t_2,j_2}\bigg\vert \boldsymbol{x}^{[1]}, \boldsymbol{x}^{[2]}\right)=\frac{ \exp(\sigma_2^2)-1}{\left\{1+\psi^{[2]}\right\}\exp(\sigma_2^2)-1};
\end{cases}
\]
for $t, t_1, t_2, j_1, j_2\in\mathbb{N}$ and $j_1\neq j_2$.

\end{example}


\section{Proofs}

\begin{proof}[Proof of Proposition \ref{prop.1}]
 Assume the Heterogeneous Model, and further assume that the frequency distribution function ${F}_1$ is in ${\it EDF}(\lambda^{[1]}\Theta^{[1]}, \psi^{[1]})$.
Then, using the conditional mean and variance method, the mean and variance of the aggregate severity conditional a priori are
\[
\begin{aligned}
\E{S_{t}\big\vert \boldsymbol{x}^{[1]}, \boldsymbol{x}^{[2]}}
&=\E{ \E{N_{t}M_{t}\big\vert N_t, \boldsymbol{x}^{[1]}, \boldsymbol{x}^{[2]},\Theta^{[1]}, \Theta^{[2]}} \big\vert \boldsymbol{x}^{[1]}, \boldsymbol{x}^{[2]} }\\
&=\E{ N_{t}  \lambda^{[2]}\Theta^{[2]}  \bigg\vert \boldsymbol{x}^{[1]}, \boldsymbol{x}^{[2]} }  \\
&=\lambda^{[1]} {\lambda^{[2]}}\E{\Theta^{[1]}\Theta^{[2]} }\\
\end{aligned}
\]
and
\[
\begin{aligned}
&\Var{S_{t}\big\vert \boldsymbol{x}^{[1]}, \boldsymbol{x}^{[2]}}\\
&\quad\quad= \E{ \Var{S_{t}\big\vert N_t, \boldsymbol{x}^{[1]}, \boldsymbol{x}^{[2]},\Theta^{[1]}, \Theta^{[2]}} \bigg\vert \boldsymbol{x}^{[1]}, \boldsymbol{x}^{[2]} }+ \Var{ \E{S_{t}\big\vert N_t, \boldsymbol{x}^{[1]}, \boldsymbol{x}^{[2]},\Theta^{[1]}, \Theta^{[2]}} \bigg\vert \boldsymbol{x}^{[1]}, \boldsymbol{x}^{[2]} }\\
&\quad\quad= \E{ \lambda^{[1]}\Theta^{[1]} V_2\left(\lambda^{[2]}\Theta^{[2]} \right) \psi^{[2]}
\Bigg\vert \boldsymbol{x}^{[1]}, \boldsymbol{x}^{[2]} }+ \left(\lambda^{[1]} \lambda^{[2]}\right)^2 \Var{\Theta^{[1]}\Theta^{[2]} }\\
&\quad\quad\quad\quad
 +\left( \lambda^{[2]} \right)^2 \E{\left( \Theta^{[2]} \right)^2 V_1\left(\lambda^{[1]} \Theta^{[1]}\right) \psi^{[1]}\Bigg\vert \boldsymbol{x}^{[1]}, \boldsymbol{x}^{[2]}}.
 \\
\end{aligned}
\]
Similarly, we calculate the covariance formulas related to severity conditional a priori as
\[
\begin{aligned}
&\cov{S_{t_1}, S_{t_2}\bigg\vert \boldsymbol{x}^{[1]}, \boldsymbol{x}^{[2]} }\\
&=\E{ \cov{N_{t_1}M_{t_1},N_{t_2}M_{t_2}\big\vert N_{t_1}, N_{t_2}, \boldsymbol{x}^{[1]}, \boldsymbol{x}^{[2]},\Theta^{[1]}, \Theta^{[2]} } \bigg\vert \boldsymbol{x}^{[1]}, \boldsymbol{x}^{[2]}  } \\
& \quad\quad +\cov{ \E{N_{t_1}M_{t_1}\big\vert N_{t_1}, N_{t_2},  \boldsymbol{x}^{[1]}, \boldsymbol{x}^{[2]},\Theta^{[1]}, \Theta^{[2]} },
    \E{N_{t_2}M_{t_2}\big\vert N_{t_1}, N_{t_2},\boldsymbol{x}^{[1]}, \boldsymbol{x}^{[2]},\Theta^{[1]}, \Theta^{[2]} }
     \bigg\vert  \boldsymbol{x}^{[1]}, \boldsymbol{x}^{[2]} }\\
&=\cov{  N_{t_1}  \lambda^{[2]}\Theta^{[2]} \, ,  N_{t_2}  \lambda^{[2]}\Theta^{[2]} \, \bigg\vert  \boldsymbol{x}^{[1]}, \boldsymbol{x}^{[2]} }\\
&=\left(\lambda^{[1]}\lambda^{[2]}\right)^2\Var{\Theta^{[1]}\Theta^{[2]} },
\end{aligned}
\]
for $t_1\neq t_2$, where the last equation comes from $ \cov{S_{t_1},S_{t_2}\big\vert N_{t_1}, N_{t_2}, \boldsymbol{x}^{[1]}, \boldsymbol{x}^{[2]},\Theta^{[1]}, \Theta^{[2]} } =0$,
and
\[
\begin{aligned}
&\cov{N_{t}, S_{t}\bigg\vert \boldsymbol{x}^{[1]}, \boldsymbol{x}^{[2]}}\\
&=\E{ \E{N_{t}S_{t}\big\vert  \boldsymbol{x}^{[1]}, \boldsymbol{x}^{[2]},\Theta^{[1]}, \Theta^{[2]} }
 \bigg\vert \boldsymbol{x}^{[1]}, \boldsymbol{x}^{[2]}  } -\E{N_{t}  \bigg\vert \boldsymbol{x}^{[1]}, \boldsymbol{x}^{[2]}  }
\E{S_{t} \bigg\vert \boldsymbol{x}^{[1]}, \boldsymbol{x}^{[2]} }   \\
&=\lambda^{[2]} \psi^{[1]}\E{{V}_1\left( \lambda^{[1]} \Theta^{[1]}\right) \Theta^{[2]}} +
\left(\lambda^{[1]}\right)^2\lambda^{[2]}\E{\left(\Theta^{[1]}\right)^2\Theta^{[2]}}
-\left(\lambda^{[1]}\right)^2\lambda^{[2]}\E{\Theta^{[1]}}\E{\Theta^{[1]}\Theta^{[2]}},
\end{aligned}
\]
where the last equality comes from
\[
\begin{aligned}
\E{N_{t}S_{t}\big\vert \boldsymbol{x}^{[1]}, \boldsymbol{x}^{[2]},\Theta^{[1]}, \Theta^{[2]} }
&= \E{ \E{N_{t},S_{t}\big\vert N_{t}, \boldsymbol{x}^{[1]}, \boldsymbol{x}^{[2]},\Theta^{[1]}, \Theta^{[2]} } \bigg\vert \boldsymbol{x}^{[1]}, \boldsymbol{x}^{[2]},\Theta^{[1]}, \Theta^{[2]}  }\\
&= \E{ N_{t}^2 \E{M_{t}\big\vert N_{t}, \boldsymbol{x}^{[1]}, \boldsymbol{x}^{[2]},\Theta^{[1]}, \Theta^{[2]} } \bigg\vert \boldsymbol{x}^{[1]}, \boldsymbol{x}^{[2]},\Theta^{[1]}, \Theta^{[2]} }\\
&=  \E{ N_{t}^2 \lambda^{[2]}\Theta^{[2]}  \bigg\vert \boldsymbol{x}^{[1]}, \boldsymbol{x}^{[2]},\Theta^{[1]}, \Theta^{[2]} }\\
&= \lambda^{[2]} \Theta^{[2]} \left\{ {V}_1\left( \lambda^{[1]} \Theta^{[1]}\right) \psi^{[1]} + \left( \lambda^{[1]} \Theta^{[1]}\right)^2 \right\}.
\end{aligned}
\]

The covariance between the frequencies conditional on a priori is
\[
\begin{aligned}
&\cov{N_{t_1}, N_{t_2}\bigg\vert \boldsymbol{x}^{[1]}, \boldsymbol{x}^{[2]}}\\
&\quad\quad=\E{ \cov{N_{t_1},N_{t_2}\big\vert \boldsymbol{x}^{[1]}, \boldsymbol{x}^{[2]},\Theta^{[1]}, \Theta^{[2]} } \bigg\vert \boldsymbol{x}^{[1]}, \boldsymbol{x}^{[2]}  } \\
& \quad\quad\quad\quad +\cov{ \E{N_{t_1}\big\vert \boldsymbol{x}^{[1]}, \boldsymbol{x}^{[2]},\Theta^{[1]}, \Theta^{[2]} },
    \E{N_{t_2}\big\vert \boldsymbol{x}^{[1]}, \boldsymbol{x}^{[2]},\Theta^{[1]}, \Theta^{[2]} }
     \bigg\vert \boldsymbol{x}^{[1]}, \boldsymbol{x}^{[2]} }\\
&\quad\quad=\left(\lambda^{[1]}\right)^2\Var{\Theta^{[1]}}
\end{aligned}
\]
for 
$t_1\neq t_2$. Finally, we can obtain the formula for $\cov{N_{t_1}, S_{t_2}\bigg\vert \boldsymbol{x}^{[1]}, \boldsymbol{x}^{[2]}}$ similarly.

\end{proof}

\begin{proof}[Proof of Theorem \ref{thm.oh.1}]
Assume the Heterogeneous Model.
First, note that $r_l^{\rm Dep}$ for $l=0, \cdots,{z}$ satisfy
\[
0=\E{-2\Lambda^{[1]} \Lambda^{[2]}(S_{T+1}-r_L^{\rm Dep} \Lambda^{[1]} \Lambda^{[2]})\big\vert L=l},
\]
which is equivalent to \eqref{eq.16}.
Here, we can calculate $\E{\left( \Lambda^{[1]}\Lambda^{[2]}\right)^2 \Theta^{[1]}\Theta^{[2]}\bigg\vert L=l}$ in \eqref{eq.16} as
\begin{equation*}
\begin{aligned}
&\E{\left( \Lambda^{[1]}\Lambda^{[2]}\right)^2 \Theta^{[1]}\Theta^{[2]}\bigg\vert L=l}\\
&=\sum\limits_{\kappa\in\mathcal{K}} \left(\lambda_{\kappa}^{[1]}\lambda_{\kappa}^{[2]}\right)^2
  \E{\Theta^{[1]}\Theta^{[2]} \big\vert L=l, \Lambda^{[1]}=\lambda_{\kappa}^{[1]}, \Lambda^{[2]}=\lambda_{\kappa}^{[2]}}\P{\Lambda^{[1]}=\lambda_{\kappa}^{[1]}, \Lambda^{[2]}=\lambda_{\kappa}^{[2]}\big\vert L=l}\\
&=\sum\limits_{\kappa\in\mathcal{K}} \left(\lambda_{\kappa}^{[1]}\lambda_{\kappa}^{[2]}\right)^2
  \int\int \theta^{[1]}\theta^{[2]} f(\theta^{[1]}, \theta^{[2]}\big\vert L=l, \Lambda^{[1]}=\lambda_{\kappa}^{[1]}, \Lambda^{[1]}=\lambda_{\kappa}^{[2]}) {\rm d}\theta^{[1]}{\rm d}\theta^{[2]}\,
  \P{\Lambda^{[1]}=\lambda_{\kappa}^{[1]}, \Lambda^{[2]}=\lambda_{\kappa}^{[2]} \big\vert L=l}\\
&=\frac{\sum\limits_{\kappa\in\mathcal{K}} w_{\kappa} \left(\lambda_{\kappa}^{[1]}\lambda_{\kappa}^{[2]}\right)^2  \int\int \theta^{[1]}\theta^{[2]}  \pi_l(\lambda_{\kappa}^{[1]}\theta^{[1]}) h(\theta^{[1]}, \theta^{[2]}){\rm d}\theta^{[1]}{\rm d} \theta^{[2]}}
{\sum\limits_{\kappa\in\mathcal{K}}w_{\kappa} \int  \pi_l(\lambda_{\kappa}^{[1]}\theta^{[1]}) g_1(\theta^{[1]}){\rm d}\theta^{[1]}  },
\end{aligned}
\end{equation*}
 where the last equation comes from the independence between the a priori and a posteriori.

Similarly, we can calculate $\E{\left( \Lambda^{[1]}\Lambda^{[2]}\right)^2\big\vert L=l} $ in \eqref{eq.16} as
\[
\begin{aligned}
\E{\left( \Lambda^{[1]}\Lambda^{[2]}\right)^2\big\vert L=l}
&=\sum\limits_{\kappa\in\mathcal{K}} \left(\lambda_{\kappa}^{[1]}\lambda_{\kappa}^{[2]}\right)^2
  \P{ \Lambda^{[1]}=\lambda_{\kappa}^{[1]}, \Lambda^{[2]}=\lambda_{\kappa}^{[2]}\big\vert L=l}\\
&=\sum\limits_{\kappa\in\mathcal{K}} \left(\lambda_{\kappa}^{[1]}\lambda_{\kappa}^{[2]}\right)^2  \frac{\P{L=l\big\vert  \Lambda^{[1]}=\lambda_{\kappa}^{[1]}, \Lambda^{[2]}=\lambda_{\kappa}^{[2]}}\P{\Lambda^{[1]}=\lambda_{\kappa}^{[1]}, \Lambda^{[2]}=\lambda_{\kappa}^{[2]}} }{\P{L=l}}\\
&=\frac{\sum\limits_{\kappa\in\mathcal{K}} w_{\kappa} \left(\lambda_{\kappa}^{[1]}\lambda_{\kappa}^{[2]}\right)^2  \int \pi_l(\lambda_{\kappa}^{[1]}\theta^{[1]})g_1(\theta^{[1]}){\rm d}\theta^{[1]}}
{\sum\limits_{\kappa\in\mathcal{K}} w_{\kappa}  \int \pi_l(\lambda_{\kappa}^{[1]}\theta^{[1]}) g_1(\theta^{[1]}){\rm d}\theta^{[1]} },
\end{aligned}
\]
 where the last equation comes from the independence between the a priori and a posteriori.
\end{proof}

\begin{proof}[Proof of Corollary \ref{cor.1}]
The proof is similar to that of Theorem \ref{thm.oh.1}, except we can represent the bivariate density function $h$ as the product of two marginal density functions.
\end{proof}

\begin{proof}[Proof of Proposition \ref{prop3}]
Assume the Heterogeneous Model and let relativity $\boldsymbol{r}$ be given. Then, we have
 \begin{equation}\label{eq.68}
 \begin{aligned}
&\E{\left(\E{S_{t+1}\big\vert \Lambda^{[1]},  \Lambda^{[2]}, \Theta^{[1]}, \Theta^{[2]}}-\Lambda^{[1]} \Lambda^{[2]} r_L\right)^2}\\
&\quad\quad\quad\quad=\sum\limits_{l=1}^{z}\sum\limits_{\kappa\in\mathcal{K}}
\E{\left(\Lambda^{[1]}\Lambda^{[2]}\Theta^{[1]}\Theta^{[2]}-\Lambda^{[1]} \Lambda^{[2]} r_L\right)^2\bigg\vert \Lambda^{[1]}=\lambda_{\kappa}^{[1]}, \Lambda^{[2]}=\lambda_{\kappa}^{[2]}, L=l}\\
&\quad\quad\quad\quad\quad\quad\quad\quad\quad\quad\quad\quad\P{\Lambda^{[1]}=\lambda_{\kappa}^{[1]}, \Lambda^{[2]}=\lambda_{\kappa}^{[2]}\big\vert L=l}\P{L=l}\\
&\quad\quad\quad\quad=\sum\limits_{l=1}^{z}\sum\limits_{\kappa\in\mathcal{K}}
\int\int \left(\lambda_\kappa^{[1]}\lambda_\kappa^{[2]} \theta^{[1]}\theta^{[2]}-
\lambda_\kappa^{[1]}\lambda_\kappa^{[2]}r_l
\right)^2\\
&\quad\quad\quad\quad\quad\quad\quad\quad\quad\quad\quad\quad
h\left(\theta^{[1]}, \theta^{[2]}\bigg\vert
\Lambda^{[1]}=\lambda_{\kappa}^{[1]}, \Lambda^{[2]}=\lambda_{\kappa}^{[2]}, L=l
\right)
{\rm d}\theta^{[1]}{\rm d}\theta^{[2]}\\
&\quad\quad\quad\quad\quad\quad\quad\quad\quad\quad\quad\quad\quad\quad\quad\quad\quad\quad\quad
\P{\Lambda^{[1]}=\lambda_{\kappa}^{[1]}, \Lambda^{[2]}=\lambda_{\kappa}^{[2]}\big\vert L=l}\P{L=l}\\
&\quad\quad\quad\quad=\sum\limits_{l=1}^{z}\sum\limits_{\kappa\in\mathcal{K}} w_\kappa
\int\int \left(\lambda_\kappa^{[1]}\lambda_\kappa^{[2]} \theta^{[1]}\theta^{[2]}-
\lambda_\kappa^{[1]}\lambda_\kappa^{[2]}r_l
\right)^2
\pi_l\left(\lambda_\kappa^{[1]}\theta^{[1]}, \psi^{[1]}\right) h(\theta^{[1]}, \theta^{[2]}) {\rm d}\theta^{[1]}{\rm d}\theta^{[2]}
\end{aligned}
 \end{equation}
 where the last equation is from the independence between the a priori and a posteriori.


\end{proof}

\section{Tables}

\begin{table}[h!]
\caption{Average frequency and severity by entity type and Year }
\centering
\resizebox{\linewidth}{!}{
\begin{tabular}{ l r r r r r r r r r r r r r r}
 \hline
  & \multicolumn{5}{l}{Average Frequency} & &  \multicolumn{5}{l}{Average Severity} \\ \cline{2-6} \cline{8-12}
  & 2006 & 2007 & 2008 & 2009 & 2010 & &  2006 & 2007 & 2008 & 2009 & 2010 \\
  \hline
 Miscellaneous 		& 0.09&	0.13&	0.15&	0.15&	0   && 12677.5&	7018.4&	 2189.4&	 1997.7&	NaN\\
 City				& 0.89&	1.02&	1.07&	0.89&	0.94&&  4622.8&	6463.5&	 3629.2&	 5289.5&	3000.9\\
 County			& 2.68&  	3.28&	4.04&	3.30&	3.39&&  9122.6&   	8108.7&	 8904.3&	 6784.9&	6321.2\\
 School			& 0.21&	0.34&	0.29&	0.20&	0.43&&  5857.0&	4593.7&	 4233.1&	 3659.4&	4232.6\\
 Town			& 0.03&  	0.05&	0.05&	0.07&	0.06&&  3222.7&	4279.2&	46020.2&	 4726.7&	3588.9\\
 Village			& 0.28&	0.32&	0.33&	0.18&	0.17&& 15656.7&	5942.0&	 5622.1&	10195.4&	6880.1\\
 \hline\hline
 Total 			& 0.52&  0.68&     0.75&   0.66&     0.76&&  8550.0&  6930.2& 7478.6& 6334.3& 5504.38 \\
 \hline
\end{tabular}
} 
\label{entity_year}
\end{table}
				
\begin{table}[h!]
\caption{Average frequency and severity by coverage and Year }
\centering
\resizebox{\linewidth}{!}{
\begin{tabular}{ l r r r r r r r r r r r r r r}
 \hline
  & \multicolumn{5}{l}{Average Frequency} & &  \multicolumn{5}{l}{Average Severity} \\ \cline{2-6} \cline{8-12}
  & 2006 & 2007 & 2008 & 2009 & 2010 & &  2006 & 2007 & 2008 & 2009 & 2010 \\
  \hline
Coverage 1& 0.08&  0.06&    0.06&    0.02&    0.09 &&  4430.2& 7920.3& 2370.4& 3474.35& 4508.81\\
Coverage 2& 0.18&  0.24&    0.32&    0.19&    0.33 && 16364.1& 4936.4& 4862.0& 4218.45& 6022.18\\
Coverage 3& 1.30&  1.73&    1.91&    1.70&    1.84 &&  7682.7& 7178.6& 8087.9& 6588.33& 5465.18\\
 \hline\hline
 Total 			& 0.52&  0.68&     0.75&   0.66&     0.76&&  8550.0&  6930.2& 7478.6& 6334.3& 5504.38 \\
 \hline
\end{tabular}
} 
\label{coverage_year}
\end{table}

\begin{table}[h!]
\centering
\caption{Observable policy characteristics used as covariates} \label{tab.x}
\begin{tabular}{l|l r r r r r r r }
\hline
Category & \multirow{2}{*}{Description} &&  \multicolumn{3}{c}{\multirow{2}{*}{Proportions}} \\
variables &  &  &  &   \\
\hline
Entity type   & Type of local government entity    \\
		& \quad\quad\quad\quad\quad\quad Miscellaneous  	&& \multicolumn{3}{c}{5.03$\%$} \\
		& \quad\quad\quad\quad\quad\quad City			&& \multicolumn{3}{c}{9.66$\%$} \\
		& \quad\quad\quad\quad\quad\quad County		&& \multicolumn{3}{c}{11.47$\%$} \\
		& \quad\quad\quad\quad\quad\quad School		&& \multicolumn{3}{c}{36.42$\%$} \\
		& \quad\quad\quad\quad\quad\quad Town			&& \multicolumn{3}{c}{16.90$\%$} \\
		& \quad\quad\quad\quad\quad\quad Village 		&& \multicolumn{3}{c}{20.52$\%$} \\
\hline
Coverage & Collision coverage amount for old and new vehicles\\
		& \quad\quad\quad\quad\quad\quad Coverage $\in (0,\quad\,\, 0.14] = 1 $   	&& \multicolumn{3}{c}{33.40$\%$} \\
		& \quad\quad\quad\quad\quad\quad Coverage $\in (0.14, \, 0.74] = 2 $		&& \multicolumn{3}{c}{33.20$\%$} \\
		& \quad\quad\quad\quad\quad\quad Coverage $\in (0.74,\quad \infty] = 3$	&& \multicolumn{3}{c}{33.40$\%$} \\
\hline
\end{tabular}
\end{table}

\begin{table}[h!]
\centering
\caption{Percentage and number of claims by frequency and policy year} \label{tab.n}
\begin{tabular}{l r r r r r r r r r }
\hline
& \multicolumn{7}{l}{Training } && \multicolumn{1}{l}{Test} \\ \cline{2-8} \cline{10-10}
Frequency & 2006 	& 2007 	& 2008 	& 2009 	& 2010	&  Count & Prop$(\%)$ && 2011 \\
\hline
0	&	329	&	307	&	304	&	304	&	284	&	1528	&	73.64	&&	288	\\
1	&	61	&	61	&	52	&	53	&	61	&	288	&	13.88	&&	46	\\
2	&	26	&	17	&	12	&	23	&	18	&	96	&	4.63	&&	12	\\
3	&	6	&	9	&	16	&	8	&	13	&	52	&	2.51	&&	14	\\
4	&	8	&	10	&	4	&	6	&	5	&	33	&	1.59	&&	6	\\
5	&	4	&	2	&	3	&	6	&	6	&	21	&	1.01	&&	6	\\
6	&	1	&	1	&	8	&	2	&	4	&	16	&	0.77	&&	1	\\
7	&	1	&	3	&	4	&	1	&	2	&	11	&	0.53	&&	2	\\
8	&		&	4	&		&	1	&		&	5	&	0.24	&&	1	\\
9	&	1	&	2	&	1	&	1	&		&	5	&	0.24	&&	1	\\
10	&		&	1	&	1	&	2	&	2	&	6	&	0.29	&&		\\
11	&		&	1	&	1	&		&	1	&	3	&	0.14	&&	2	\\
12	&	2	&		&		&		&		&	2	&	0.10	&&		\\
13	&		&	1	&	1	&		&	1	&	3	&	0.14	&&		\\
14	&		&		&		&		&	1	&	1	&	0.05	&&		\\
16	&		&		&	1	&	1	&		&	2	&	0.10	&&		\\
18	&		&		&	1	&		&		&	1	&	0.05	&&		\\
21	&		&		&		&	1	&		&	1	&	0.05	&&		\\
22	&		&		&		&		&	1	&	1	&	0.05	&&		\\
\hline
Count&	439	&	419	&	409	&	409	&	399	&	2075	&	100	&&	379	\\
\hline \hline
\end{tabular}
\end{table}

\begin{table}[h!]
\centering
\caption{Percentage and number of claims by frequency and policy year} \label{tab.m}
\begin{tabular}{l r r r r r r r r }
\hline
& \multicolumn{6}{l}{Training } && \multicolumn{1}{l}{Test} \\ \cline{2-7} \cline{9-9}
Frequency & 2006 	& 2007 	& 2008 	& 2009 	& 2010	& Avg.Sev. && 2011 \\
\hline
1	&	10605&	6670	&	6969	&	5111	&	3843	&	 6672&&  5394\\
2	&	4974	&	6864	&	6158	&	6795	&	11186&	 7058&&  3700\\
3	&	16910&	9248	&	10741&	10579&	6777	&	10179&&  6785\\
4	&	9315	&	12667&	8594	&	3486	&	9329	&	 9186&&  4037\\
5	&	6472	&	4765	&	8556	&	9870	&	5702	&	 7358&& 10941\\
6	&	1583	&	1361	&	5394	&	5402	&	2579	&	 4201&&  2389\\
7	&	10272&	5132	&	8143	&	3566	&	2205	&	 6020&&  3685\\
8	&		&	4479	&		&	1827	&		&	 3949&&  6316\\
9	&	14784&	5638	&	6288	&	10006&		&	 8471&&  3836\\
10	&		&	3249	&	12356&	3645	&	6708	&	 6052&&  1550\\
11	&		&	8542	&	3205	&		&	684	&	 4144&&	\\
12	&	4419	&		&		&		&		&	 4419&&	\\
13	&		&	2287	&	8286	&		&	2416	&	 4329&&	\\
14	&		&		&		&		&	2931	&	 2931&&	\\
16	&		&		&	3010	&	7942	&		&	 5476&&	\\
18	&		&		&	8534	&		&		&	 8534&&	\\
21	&		&		&		&	4701	&		&	 4701&&	\\
22	&		&		&		&		&	4882	&	 4882&&	\\
\hline
AvgSev&	8550	&	6930	&	7479	&	6334	&	5504	&	6891&& 5476\\
\hline \hline
\end{tabular}
\end{table}


\end{document}